\documentclass[11pt,a4paper,leqno]{article}

\usepackage{amsmath}
\usepackage{amsfonts}
\usepackage{amssymb}
\usepackage{amsthm}

\usepackage{graphicx}
\usepackage{xcolor}
\usepackage{url}
\usepackage{float}

\usepackage[english]{babel}
\usepackage[utf8]{inputenc}

\usepackage[T1]{fontenc}

\numberwithin{equation}{section}

\theoremstyle{plain}

\newtheorem{theorem}[subsection]{Theorem}
\newtheorem{proposition}[subsection]{Proposition}
\newtheorem{lemma}[subsection]{Lemma}
\newtheorem{corollary}[subsection]{Corollary}
\newtheorem{conjecture}[subsection]{Conjecture}

\theoremstyle{definition}

\newcommand{\Q}{\mathbb{Q}}
\newcommand{\Z}{\mathbb Z}
\newcommand{\C}{\mathbb C}

\newcommand{\R}{\mathbb{R}}

\title{Plateaux of probability for the expanded quantum infinite well}
\author{Fernando Chamizo
\and
Dulcinea Raboso
\and
Osvaldo P. Santill\'an}
\date{\today}

\begin{document}
\pagestyle{plain}

\maketitle

\begin{abstract}
 If the standard 1D quantum infinite potential well initially in its ground state suffers a sudden expansion, it turns out that in the evolution of the system they may appear plateaux of probability for some fractional times, as noticed by C. Aslangul in 2008. We introduce a mathematical framework to explain this phenomenon. Remarkably, the characterization of these plateaux depends on nontrivial number theoretical considerations.
\end{abstract}

\section{Introduction}
In \cite{aslangul} it is 
considered the basic quantum mechanical setting of the infinite potential well in $[0,a]$
imposing as initial condition the first eigenstate, $$\Psi_0=\sqrt{\frac{2}{a}}\,\sin\frac{\pi x}{a},$$ and a sudden expansion of the well to $[0,\Lambda a]$ (with a fixed $\Lambda>1$) is modeled  preserving this initial condition in $[0,a]$ and putting $\Psi(x,0)=0$ in the rest of the interval $[0,\Lambda a]$.
Changing the units, 
it is clear that one can assume $a=1$, as it is done here in the sequel. Let us generalize a bit the situation admitting the $N$-th eigenstate as initial condition before the expansion. The Schr\"odinger equation ruling the wave function after the expansion is
\begin{equation}\label{schr}
 \begin{cases}
  i\hbar \partial_t\Psi(x,t)=
  -\dfrac{\hbar^2}{2m}\partial_{xx}\Psi(x,t)\quad\text{for $0<x<\Lambda $, \quad$t>0$,}
  \\[6pt]
  \Psi(0,t)=\Psi(\Lambda ,t)=0,\quad \Psi(x,0)=\sqrt{2}\,\theta(1-x)\sin(N\pi x)
 \end{cases}
\end{equation}
where $\theta(x)$ is the Heaviside step function taking the value $1$ if $x\ge 0$ and $0$ if $x<0$.

Using standard techniques, the solution of this equation admits the series expansion
\begin{equation}\label{formula}
 \Psi(x,t)
 =
 \frac{i\Lambda \sqrt{2}}{\pi}
 \sum_{n\in\Z}
 \frac{\sin(\pi n/\Lambda)}{n^2-\Lambda^2}
 e
 \Big(
 \frac{nx}{2\Lambda}-\frac{n^2t}{T}
 \Big)
 \qquad\text{with}\quad
 T=\frac{4m\Lambda^2}{\pi \hbar}
\end{equation}
where here and thereafter it is employed the notation
\[
 e(x)=e^{2\pi ix}
\]
to avoid the continual writing of $2\pi i$ exponents.
If $\Lambda$ is a positive integer then the troublesome~$0/0$ coefficients must be understood as a limit $\Lambda\to |n|$.
\medskip

In \cite{aslangul} it is observed numerically the surprising fact consisting in that for some fractional multiples of $T$ and some values of $\Lambda$, the probability density $|\Psi|^2$ possesses plateaux: It is constant as a function of $x$ in some intervals. In fact, a complete theoretical explanation is provided when $t=T/2$, $t=T/4$ and $t=T/8$ via explicit formulas for $\Psi(x,t)$ for these values of $t$. A number of conjectures are made for other fractional values.
\bigskip

The scheme and purpose of the present paper is as follows. 
In \S2 it is shown that these conjectures can be proved and, in fact, there are explicit formulas for arbitrary fractional times under general initial conditions revealing that 
the Heaviside step function 
for $t=0$ in \eqref{schr} causes
$0$-level plateaux (intervals carrying zero probability) in the solution for $\Lambda$ beyond a certain threshold. Following \cite{aslangul}, this is named the fragmentation case. 
The argument is elementary and general, it is an instance of the quantum Talbot effect. It was applied  in \cite[\S5]{ChSa} to evaluate explicitly a wave function related to quantum revivals on the sphere and it has been employed in other contexts related to dispersive equations. 
For instance, in \cite{olver} the idea is presented in general simple terms and in \cite{OlTs} the case of the linear Korteweg-De Vries equation (a third order dispersive equation) is studied in detail.

On the other hand, the plateaux carrying positive probability for the problem \eqref{schr} are more mysterious and do not parallel the standard theory on the Talbot effect. They are linked to the non fragmentation case and they actually appear, as exemplified in figures~2 and~3 of \S5 depicting the probability function for some choices of the parameters. 
A first step to characterize the plateaux is given in \S3 where it is provided a general analytic criterion to decide if given $x$ there is a plateau in a neighborhood of $x$. 

As suggested before, the fragmentation case is by far simpler. The result in \S4 lists all the plateaux appearing in this case, which correspond to equally spaced segments. This can be obtained easily with the tools of  \S2 and it is also a consequence of the criterion in \S3. 

Our approach to study the non fragmentation case requires some knowledge on cyclotomy that is introduced in \S5. Historically, cyclotomy was the starting point of Galois theory in number fields. Gauss introduced it to determine the regular polygons that can be constructed with ruler and compass. Since then, it has been a recurrent topic in algebraic number theory. 

The main results in the paper are in \S6 where the non fragmentation case is approached. It is proved that for $2N\Lambda$ an odd integer, there is a plateau and it is conjectured that there do not exist plateaux in other cases. This conjecture is proved when the denominator of $t/T$ is prime. Two additional results are stated providing support to the conjecture.

Finally, in \S7 there are some figures corresponding to numerical simulations to exemplify the results. The code to generate them is freely available in 
{\url{https://matematicas.uam.es/~fernando.chamizo/dark/d_plateaux.html}}.
\medskip

It is noticeable how the problem \eqref{schr}, which is a variation of the infinite potential well appearing in any quantum mechanics undergraduate course, leads to highly nontrivial considera\-tions in number theory.  

\section{A general argument and some consequences}

For $\{a_n\}_{n\in\Z}$ such that $\sum_{n\in\Z}|a_n|<\infty$, consider the real variable functions $f: \R^2\longrightarrow\C$ and $g: \R\longrightarrow\C$ given by
\[
 f(x,t)
 =
 \sum_{n\in\Z} a_n e(nx-n^2t)
 \qquad\text{and}\qquad
 g(x)=f(x,0).
\]
The reader will have no trouble in recognizing that this describes the evolution of a free particle with the function $g(x)$ as initial condition. An elementary argument embodied in the following result allows to express $f$ at fractional times in terms of $g$.

\begin{lemma}\label{main}
 Let $a/q$ be an irreducible fraction (with $q\in\Z^+$) then
 \[
  f\big(x,\frac aq\big)
  =
  \frac{1}{q}
  \sum_{k=0}^{q-1}
  G^*(a,k,q) g\big(x+\frac{k}{q}\big)
  \qquad\text{with}\qquad
  G(a,k,q)
  =
  \sum_{\ell=0}^{q-1}
  e\Big(
  \frac{a\ell^2+k\ell}{q}
  \Big).
 \]
\end{lemma}
Here the asterisk indicates the complex conjugation. the notation $G(a,k,q)$ is taken from number theory where these sums are known as (generalized) \emph{quadratic Gauss sums}.  It is not difficult to 
understand how this formula is obtained.   Substituting the definitions in the sum in the statement leads to
\[
 \sum_{k=0}^{q-1}
 G^*(a,k,q) g\big(x+\frac{k}{q}\big)
 =
 \sum_{k=0}^{q-1}
 \sum_{n\in\Z}
 a_n e(nx) e\Big(\frac{nk}{q}\Big)
 \sum_{\ell=0}^{q-1}e\Big(-\frac{a\ell^2+k\ell}{q}\Big).
\]
The last sum  can be rearranged as
\[
 \sum_{n\in\Z}
 a_n e(nx)
 \sum_{\ell=0}^{q-1}e\Big(-\frac{a\ell^2}{q}\Big)
 \sum_{k=0}^{q-1}
 e\Big(\frac{(n-\ell) k}{q}\Big)
 =
 q
 \sum_{n\in\Z}
 a_n e(nx)
 \sum_{\ell=0}^{q-1}e\Big(-\frac{a\ell^2}{q}\Big)\delta_{n\ell}
\]
and the result follows.
\medskip

Although the exact evaluation of $G(a,k,q)$ is intricate and depends on arithmetic topics related to quadratic reciprocity, it is fairly easy to get a formula for its absolute value. For further reference, it is:
\begin{equation}\label{gav}
 \big|G(a,k,q)\big|
 =
 \begin{cases}
  \sqrt{q} &\text{if $q$ is odd,}
  \\
  \sqrt{2q} &\text{if $q$ and $k+\frac q2$ are even,}
  \\
  0 &\text{otherwise.}
 \end{cases}
\end{equation}
The proof reduces to some elementary manipulations in the double sum arising when expanding $G^*(a,k,q)G(a,k,q)$ (see the details in \cite[\S4]{ChSa}).  Concerning the phase of this complex number, it may be proven the following result. The proof provided here has the advantage that  only appeals to elementary known arguments.

\begin{lemma}\label{mains}
The Gauss sums values are 
 \begin{equation}\label{garg}
  G^*(a,k,q)=\sqrt{q}\, e^{i\alpha(a,q)}c_{a/q}(k)
 \end{equation}
 for certain $\alpha(a,q)\in\R$ not depending on $k$ and 
\begin{equation}\label{coef}
  c_{a/q}(k)=
  \begin{cases}
   e\Big(\frac{\overline{4a} k^2}{q}\Big)
   &
   \text{if $q$ is odd},
   \\
   0
   &
   \text{if $q$ is even and $k+\frac q2$ is odd,}
   \\
   \sqrt{2}\,  e\Big(\frac{\overline{a} k^2}{4q}\Big)
   &
   \text{if $q$ and $k+\frac q2$ are even}.
  \end{cases}
\end{equation}
Here the bar indicates the inverse modulo $q$ i.e., for $a$ and $q$ coprime,  $\overline{a}$ is an integer $0<\overline{a}<q$ such that $q$ divides $a\overline{a}-1$.
\end{lemma}

The condition of being coprime is known to be  necessary and sufficient for the inverse to exist. Note that this is the situation considered here and thus this inverse is well defined.

\begin{proof}
If $q$ is odd, then $q$ and $2$ are coprimes and thus $\overline{2}$ exists. By completing squares
 $$
 -(an^2+kn)=-a(n+\overline{2a}k)^2+\overline{4a}k^2\qquad  \text{modulo}\;q.
 $$
  This implies $G^*(a,k,q)=G^*(a,0,q)c_{a/q}(k)$ for $q$ odd. By \eqref{gav}, $\big|G(a,0,q)\big|=\sqrt{q}$ and \eqref{garg} follows. 
 
 If $q$ is even, $G(a,k,q)$ vanishes for $k+q/2$ odd and \eqref{garg} becomes trivial. If $k+q/2$ is even, say $2N_k$, then 
 $$-(an^2+kn)=-a(n+\overline{a}N_k)^2-\frac q2(n+\overline{a}N_k) -\frac{\overline{a}q^2}{16}+\frac{\overline{a}k^2}{4}\qquad \text{modulo}\;q.$$ 
The proof of this claim reduces to expand the latter expression. This implies that
 $$G^*(a,k,q)=G^*\big(a,\frac q2,q\big)e\big(-\frac{\overline{a}q^2}{16}\big)2^{-1/2} c_{a/q}(k),$$ and the factor multiplying to $c_{a/q}(k)$ does not depend on~$k$ and its absolute value is $\sqrt{q}$ by \eqref{gav}, as expected.
\end{proof}

 The proof given above could be replaced by the full evaluation of the Gauss sums but it would be cumbersome and leading to discuss a number of  cases (cf. \cite{HaBe}).

It is direct to see that the infinite series \eqref{formula} describing the wave equation \eqref{formula} solving \eqref{schr} 
simplified as follows for fractional times
\begin{equation}\label{formula2}
 \Psi\big(2\Lambda x,\frac aq T\big)
 =
 f\big(x,\frac aq\big)=
 \frac{\sqrt{2}}{q}
 \sum_{k=0}^{q-1}
  G^*(a,k,q) g\Big(x+\frac{k}{q}\Big)
\end{equation}
where $g(x)=\theta(1-2\Lambda x)\sin(2\pi N\Lambda x)$
for $0\le x\le 1/2$ and it is extended to an odd $1$-periodic function.
A convenient compact alternative expression for $g$ is 
\begin{equation}\label{altg}
 g(x)=
 \begin{cases}
  \sin\big(2\pi N\Lambda(x-x_c)\big)
  &
  \text{if }
  |x-x_c|\le \frac{1}{2\Lambda},
  \\
  0
  &
  \text{if }
  \frac{1}{2\Lambda}<|x-x_c|
 \end{cases}
 \qquad\text{with $x_c$ the nearest integer to $x$.}
\end{equation}
In these terms, it is possible to anticipate the existence of plateaux in certain cases, which is one of the main tasks of the present paper.
The presence of plateaux  follows from the fact that $g(x)$ is supported on an interval of length $\Lambda^{-1}$ in each period and the $q$ translations in the sum  \eqref{formula2} cover a set of measure $\Lambda^{-1}q$. Hence if $\Lambda>q$ the sum must vanish in certain part of each unit interval and there are 0-level plateaux in $[0,\Lambda]$ in which $\Psi\big(x,\frac{aT}{q}\big)=0$. With the language in \cite{aslangul}, this proves that $q$ is a threshold for $\Lambda$ to have ``fragmentation''. This threshold is sharp for $q$ odd because $G(a,k,q)\ne 0$ by \eqref{gav}. If $q$ is even, again by \eqref{gav}, only one half of the terms in \eqref{formula2} contribute to the sum, those with $k$ and $q/2$ having the same parity. Then the fragmentation appears for $\Lambda>q/2$ when $q$ is even.

According to the previous considerations, if $\Lambda>q$ for $q$ odd or $\Lambda>q/2$ for $q$ even the probability density ``fragmentizes'' into blocks of $2N$ peaks (coming from the peaks of $|g|$) separated by forbidden zones  ($0$-level plateaux). The blocks in the extremes can be halved. 
If $q$ is odd, there are $qN$ peaks in total and if $q$ is even, there are $qN/2$. This situation is named fragmentation for obvious reasons.
One of the purposes of the present work is to prove the Theorem \ref{fragmen} given below, which characterizes the exact location of these forbidden zones. In addition some explicit examples are presented in \S5. 

The existence of plateaux was already anticipated in some cases situations in reference \cite{aslangul}. 
Part of its derivation follows from the identities to be described below. It may be an instructive excercise to derive them in terms of the Gauss sums. With this purpose in mind, it is convenient to express the formula \eqref{formula}, 
by making the redefinition  $x\mapsto \frac{x}{2\Lambda}$ and $t\mapsto \frac{t}{T}$ to define
 \[
  F(x,t)=
  f\Big(\frac{x}{2\Lambda},\frac tT\Big)=
 \sum_{n\in\Z} a_n e\Big(\frac{nx}{2\Lambda}-\frac{n^2t}T\Big).
 \]
With this redefinition, the lemma \ref{main} applied to the cases $a/q=1/2,\, 1/4,\, 1/8$ leads to the explicit formulas, which were already found in reference \cite{aslangul} for $t=T/2,\, T/4,\, T/8$. The only assumption is the anti-symmetry of the coefficients, that is, $a_n=-a_{-n}$. 
It is an instructive exercise using the Gauss sums, for this reason it is presented here.

\begin{corollary}\label{exp248}
 If $a_n=-a_{-n}$ then the following relations
 \[
  F(x,\frac{T}{2})=-F(\Lambda-x,0),
 \]
 \[
  F(x,\frac{T}{4})=
  \frac{1-i}{2}F(x,0)
  -
  \frac{1+i}{2}F(\Lambda-x,0),
 \]
 \[
  F(x,\frac{T}{8})=
  \frac{1-i}{2\sqrt{2}}F(x,0)
  +
  \frac{1}{2}F(x+\frac{\Lambda}{2},0)
  +
  \frac{1-i}{2\sqrt{2}}F(\Lambda-x,0)
  -
  \frac{1}{2}F(\frac{\Lambda}{2}-x,0).
 \]
 are valid.
\end{corollary}

The proof of Corollary~\ref{exp248} reduces to apply Lemma~\ref{main}  using that  $a_n=-a_{-n}$ assures that $F(x,0)$ is an odd function.

\begin{proof}
Choosing $q=2$,  $G(1,k,2)=1+e^{\pi i (1+k)}=1-(-1)^k$.
Hence, by Lemma~\ref{main}, $$F(x,\frac{T}{2})=\frac 22 F(x+\Lambda,0),$$ which is equal to $-F(\Lambda-x,0)$ because $F(x,0)$ is odd and $2\Lambda$-periodic.

If $q=2^N$ with $N>1$, $G(a,k,q)=0$ for $k$ odd since \eqref{gav}. A calculation shows $$G(1,k,4)=2+2(-1)^{k/2} i,$$ for $k=0,2$. Then by Lemma~\ref{main}
\[
 F(x,\frac{T}{4})=
 \frac 14 (2-2i)F(x,0)
 +
 \frac 14 (2+2i)F(x+\Lambda,0)
\]
and $F(x+\Lambda,0)= -F(\Lambda-x,0)$ as before. 

The case $q=8$ is similar. By  using that 
$G(1,k,8)=4$ for $k=2,6$ and $$G(1,k,8)=2(-1)^{k/4}\sqrt{2}(1+i)$$ for $k=0,4$, which gives in Lemma~\ref{main}
\[
  F(x,\frac{T}{8})=
  \frac{1-i}{2\sqrt{2}}F(x,0)
  +
  \frac{1}{2}F(x+\frac{\Lambda}{2},0)
  -
  \frac{1-i}{2\sqrt{2}}F(x+\Lambda,0)
  +
  \frac{1}{2}F(x+\frac{3\Lambda}{2},0).
\]
This shows the last equality and completes the proof.
\end{proof}

In \cite[\S4]{aslangul} these relations look more complicated when applied to \eqref{schr} because $F(x,0)$ is not considered as an $2\Lambda$-periodic function and consequently some translations must be introduced to move the arguments to the domain $[0,\Lambda]$ in which \eqref{schr} is posed.  
\medskip

The conjectures in \cite[\S4.3]{aslangul} follow from Lemma~\ref{main} with similar arguments. The following sections are intended to generalize the results of this reference, by use of algebraic and number theoretical methods.

\section{A criterion for the probability density plateaux}
The plateaux  can be seen by studying the probability distribution rather than the wave function itself.
To uniformize the ranges, instead of considering the probability density $|\Psi|^2$ for the problem \eqref{schr}, it is convenient to  re-scale the spatial variable as in \eqref{formula2}.  For an expansion factor $\Lambda$ and a fractional time $aT/q$, consider the normalized probability density
\[
 p(x)
 =
 2\Lambda
 \big|\Psi(2\Lambda x,\frac{aT}{q})\big|^2
 \qquad\text{with}\quad 
 x\in [0,\frac{1}{2}].
\]
The initial $2\Lambda$ factor is introduced to force the total probability condition $\int_0^{1/2} p(x)dx=1$. 

Let us identify a plateau for $p(x)$ with a maximal sub-interval of $[0,\frac{1}{2}]$ in which $p$ is constant. 
Since the considerations in the last section, for $\Lambda$ large, $p(x)$ presents $0$-level plateaux, intervals in which it vanishes. If this is not the case, the numerical experiments show that for $N=1$ and half-integral $\Lambda$, the common situation is the existence of a unique plateau in which $p$ does not vanish. This uniform probability distribution in an interval may sound counterintuitive   because the sine function  coming from the initial condition represented by $g$ in \eqref{formula2} is expected to induce
oscillations everywhere. The following sections are devoted to explain this non expected situation.
\medskip

The goal of this section is to establish a criterion to know if a value of $x$ belongs to a plateau.
This criterion is completely based on the following result that re-writes \eqref{formula2} in a more manageable form. 

\begin{lemma}\label{bformula}
 With the previous notation, 
 \[
  p(x)
  =
  \frac{4\Lambda}q
  \bigg|
  \sum_{k\in I(x)}
   c_{a/q}(k) \sin\Big(2\pi N\Lambda\Big(x-\frac{k}{q}\Big)\Big)
  \bigg|^2.
 \]
 where the interval $I(x)$ given by
 $$
 I(x)
 =
 \Big[
 -\frac{q}{2\Lambda}+xq,
 \frac{q}{2\Lambda}+xq
 \Big].
 $$
 was introduced and the coefficients $c_{a/q}(k)$ are defined in \eqref{coef}.
\end{lemma}
Before presenting the proof, let us say that a value $0<x<\frac{1}{2}$ is \emph{nonsingular} if $q$ is odd and both extreme point of~$I(x)$ are not integers or if $q$ is even and both extreme points plus $\frac{q}{2}$ are not even integers. This condition assures that the  sum of $c_{a/q}(k)f(k)$, with $f$ arbitrary, over the set $I(x)\cap \Z$ is preserved under small variations of $x$ because $I(x)$ does not capture new nonzero terms. The proof then goes as follows.

\begin{proof}
 
 Once \eqref{garg} is proved, the definition of $p(x)$ and \eqref{formula2} give
 \[
  p(x)
  =
  \frac{4\Lambda}{q}
  \Big|
  \sum_{k=0}^{q-1}
  c_{a/q}(k) g(x+\frac{k}{q})\Big|^2
  =
  \frac{4\Lambda}{q}
  \Big|
  \sum_{k=0}^{q-1}
  c_{a/q}(k) g(x-\frac{k}{q})\Big|^2
 \]
 where the second equality comes from the change $k\mapsto -k\pmod{q}$ using that $c_{a/q}(k)$ remains invariant under it. Note that $g(x-\frac{k}{q})$ is well defined for $k$ modulo $q$ because of the $1$-periodicity of $g$. This periodicity also assures that $x-\frac{k}{q}$ can be replaced by $x-\frac{k}{q}-y_k$ where~$y_k$ is the nearest integer to $x-\frac{k}{q}$. After these replacements, by recalling \eqref{altg},  it is deduced that
 \[
  p(x)
  =
  \frac{4\Lambda}{q}
  \Big|
  \sum_{\substack{k=0\\ \big|x-\frac kq-y_k\big|\le \frac{1}{2\Lambda}}}^{q-1}
  c_{a/q}(k) \sin\big(2\pi N\Lambda (x-\frac{k}{q}-y_k)\big)\Big|^2.
 \]
 The proof is completed using the change $k\mapsto k-qy_k$ that leaves $c_{a/q}(k)$ invariant and by realizing that
 that the condition $k\in I(x)$ is equivalent to $|x-\frac kq|\le \frac{1}{2\Lambda}$. 
\end{proof}

After the above lemma has been proved, here goes the criterion for classifying the positions where the plateaux are located. 
\begin{proposition}\label{criterion}
 Let
 \[
  S_{\pm}(x)=
  \sum_{k\in I(x)}
   c_{a/q}(k) e\Big(\pm \frac{N\Lambda k}{q}\Big)
 \]
 where $k$ takes integer values. 
 If $x\in (0,1/2)$ is non singular then the probability density~$p$ 
 is constant in a neighborhood of $x$ if and only if 
 $S_{+}(x)=0$ or $S_{-}(x)=0$. In the first case the constant is $\frac{\Lambda}{q} \big|S_{-}(x)\big|^2$ and in the second it is $\frac{\Lambda}{q} \big|S_{+}(x)\big|^2$.
\end{proposition}

\begin{proof}[Proof of Proposition~\ref{criterion}]
 As pointed out before, if $x$ is non singular then there exists an open interval $J$ around $x$ such that the set $I(y)\cap \Z$ when $y$ varies in $J$ captures the same integers giving a nonzero contribution to $S_{\pm}(y)$. Hence $S_{\pm}$ are constant in $J$ and Lemma~\ref{bformula} implies
 \[
  p(y)=
  \frac{\Lambda}{q}
  \big|
  S_{-}(x)e(2N\Lambda y)-S_{+}(x)\big|^2
  \qquad\text{for}\quad y\in J,
 \]
the last last equality follows from the elementary identity
$\sin(2\pi t)= \frac{e(-t)}{2i}\big(e(2t)-1\big)$
 If either $S_{-}(x)$ or  $S_{+}(x)=0$ then clearly $p$ is constant in $J$ and the last claim in the statement follows. On the other hand, if $S_{-}(x)S_{+}(x)\ne 0$ then $p$ cannot be locally constant because  the last formula can be interpreted as the distance from a point  $\big(|S_{+}(x)|,0\big)$ to an arc of the centered circle of radius $|S_{-}(x)|$.
\end{proof}
It should be emphasized that last criterion will be essential in the presentation of the following results.

\section{The characterization of the plateaux for the fragmentation case}

Borrowing the terminology from \cite{aslangul}, let us say that there is \emph{fragmentation} if $\Lambda>q$ when $q$ is odd or if $\Lambda>q/2$ when $q$ is even. Using \eqref{formula2} and
following the comments at the end of \S2 it is not difficult to characterize the plateaux in this case. This is collected in the following result.

\begin{theorem}\label{fragmen}
 If there is fragmentation, all the plateaux correspond to zero probability density (forbidden zones). They are intervals centered at $c_m$ with radius $r$ where
 
 \textrm{\rm a)} For $q$ odd, $c_m=\frac 1q(m+1/2)$ with $0\le m<\frac{1}{2}(q+1)$ and $r=\frac{1}{2q}-\frac{1}{2\Lambda}$.
 
 \textrm{\rm b)} For $q$ multiple of $4$,  $c_m=\frac 2q(m+1/2)$ with $0\le m<\frac{1}{4}q$ and $r=\frac{1}{q}-\frac{1}{2\Lambda}$.
 
 \textrm{\rm c)} For $q-2$ multiple of $4$,  $c_m=\frac 2q m$ with $0\le m<\frac{1}{4}(q+2)$ and $r=\frac{1}{q}-\frac{1}{2\Lambda}$.
\end{theorem}

The end intervals are clipped to adjust them to $[0,1/2]$. This clipping only affects to the last interval in a), corresponding to $m=\frac{1}{2}(q+1)$, and to the first interval in c), corresponding to $m=0$. Both are halved, they become 
$\big[\frac 12+\frac{1}{2\Lambda}-\frac{1}{2q},\frac 12\big]$
and 
$\big[0,\frac{1}{q}-\frac{1}{2\Lambda}\big]$, respectively.
\medskip

Just for illustration, let us check how to obtain Theorem~\ref{fragmen}~a) from the criterion in Proposition~\ref{criterion}.  Recalling the definition $I(x) =
 \Big[
 -\frac{q}{2\Lambda}+xq,
 \frac{q}{2\Lambda}+xq
 \Big],$
it is clear that  the fragmentation condition $\Lambda>q$ for $q$ odd implies $|I(x)|<1$.  This means that $I(x)$ contains at most an integer. The proposition~\ref{criterion} shows that a null probability density is characterized by $I(x)=\emptyset$. From the definition of $I(x)$, this is equivalent to impose that the distance of $qx$ to any integer is greater than $\frac{q}{2\Lambda}$. If $m$ is the integral part of $qx$, this means $m+\frac{q}{2\Lambda}<qx<m+1-\frac{q}{2\Lambda}$ and with the notation of the statement, this is the same as $x\in (c_m-r,c_m+r)$. 
The items b) and c) could be obtained in the same way.

In the critical case $\Lambda=q$ when $q$ is odd or $\Lambda=q/2$ when $q$ is even, it is easy to see that $S_{+}(x)$ and $S_{-}(x)$ reduce to a single complex exponential in Proposition~\ref{criterion}, for $x$ nonsingular, and hence there are not plateaux. 

If $\Lambda<q$ when $q$ is odd or if $\Lambda<q/2$ when $q$ is even, we say that there is not fragmentation. This is the difficult case to analyze.

The previous section established the criterion \ref{criterion} for the existence of the plateaux. However, there are further aspects to be analyzed. First, one should determine whether the plateaux correspond to zero or finite probability, if these plateux are unique or not, and how large are them. In addition, the non fragmentation case has to be described further, which is the harder case. The following sections are devoted to these questions.

\section{A result on cyclotomy}

Before entering into the analysis of the nonfragmentation case, it is important present some preliminary results about polynomials. These propositions are well known in algebra, but the intention is to present them here in terms of  the most elementary possible explanations. These are in fact the most technical but simultaneously more interesting tools to be employed along the present work. These results can be readily obtained with basic Galois theory. However, in order to minimize the prerequisites, it is presented here as a consequence of the irreducibility of the cyclotomic polynomial, which is more elementary. 

Given a positive integer $N$  consider the set of integers $n<N$ with has no common divisors with $N$ except the unity. This set will be denoted in the following as
\[
 \mathcal{R}_N=
 \big\{0<n\le N\,:\, \gcd(n,N)=1
 \big\}.
\]
Given this set, the $N$-th cyclotomic polynomial is defined by 
\begin{equation}\label{ciclotimico}
 C_N(x)=\prod_{n\in\mathcal{R}_N}
 \big(x-\zeta^n\big)
 \qquad\text{with}\quad 
 \zeta= e\big(\frac{1}{N}\big). 
\end{equation}

\begin{lemma}\label{irreduci}
The polynomial $C_N(x)$ has rational coefficients, actually integral\footnote{This is a straightforward from $x^N-1= \prod C_{d}(x)$ where $d$ ranges over the divisors of $N$.},  and it is irreducible in $\Q[x]$.
\end{lemma}

Recall that irreducibility in $\Q[x]$ means that the $C_N(x)$
can not be expressed as a product of two non constant polynomials in $\Q[x]$. This is a classic result that admits several elementary, albeit deep, proofs \cite{weintraub}. This fact was already pointed out by Gauss, who gave a beautiful (but involved) proof for prime $N$ when introducing cyclotomy  in his masterpiece \emph{``Disquisitiones arithmeticae''} published in 1801.  

One important application of the cyclotimic polynomials are captured from the following lemma, applied to $\Q[x]$.

\begin{lemma}\label{cyclo2}
Consider a polynomial $C(x)$ which is irreducible in $\Q[x]$. Then any polynomial  $Q(x)$ also in $\Q[x]$, which shares a at least on root with $C(x)$, is divisible by $C(x)$. In other words, all the roots of $C(x)$ are roots of $Q(x)$.
\end{lemma}

If this lemma is applied to the irreducible polynomial $C_N(x)$, then it follows that any polynomial $Q(x)$ which shares a single root with a $C_N(x)$, shares all their roots. In other words $Q(\zeta^n)=0$ if $n\in\mathcal{R}_{N}$. This allows to prove powerful results for polynomials  such as the following.

\begin{proposition}\label{cyclo}
Consider $P(x,y)$ a polynomial in two variables with rational coefficients.  Then given the following roots of unity  $\zeta_1=e(\frac{1}{n_1})$  and  $\zeta_2=e(\frac{1}{n_2})$  with $n_1$ and $n_2$ positive integers, if $P(\zeta_1,\zeta_2)=0$  it follows that
 $$
 P\big(\zeta_1^{m_1},\zeta_2^{m_2}\big)=0
 $$
 for any  ${m_1}\in\mathcal{R}_{n_1}$ and ${m_2}\in\mathcal{R}_{n_2}$ satisfying that ${m_1}-{m_2}$ is a multiple of $\gcd(n_1,n_2)$. 
\end{proposition}

\begin{proof}
 The proof relies partially on Lemma~\ref{cyclo2}. Denote $g=\gcd(n_1,n_2)$ and construct in terms in terms of the polynomial $P(x,y)$ the following single variable polynomial  $Q(x)$ 
 $$
Q(x)=P\big(x^{n_2/g},x^{n_1/g}\big).
 $$
 It is clear that $Q(x)\in\Q[x]$. In addition, the value $x=\zeta= e(\frac{1}{N})$ with $N=n_1 n_2 g^{-1}$ converts the last expression into
 $$
Q\Big(e\big(\frac{1}{N}\big)\Big)=P\Big(e\big(\frac{1}{n_1}\big), \;e\big(\frac{1}{n_2}\big)\Big)=P(\zeta_1,\zeta_2)=0,
 $$
 the last equality is simply one of  the hypothesis of the proposition. This means that $\zeta= e(\frac{1}{N})$ is a root of $Q(x)$. On the other hand, $\zeta$ is a zero of $C_N(x)$. Therefore Lemma \ref{cyclo2} implies that $C_N(x)$ divides $Q(x)$
 and furthermore for $n\in \mathcal{R}_{N}$,  it is deduced that
$$Q(\zeta^n)=P\big(\zeta_1^n,\zeta_2^n\big)=0,$$ since $C_N(\zeta^n)=0$.  In these terms the desired property $P\big(\zeta_1^{m_1},\zeta_2^{m_2}\big)=0$ will follow if it can be proved that
 $n\equiv {m_1}\pmod{n_1}$, $n\equiv {m_2}\pmod{n_2}$ since in this case $\zeta_1^n=\zeta_1^{m_1}$ and $\zeta_2^n=\zeta_2^{m_2}$. The hypothesis that ${m_1}-{m_2}$ is a multiple of $\gcd(n_1,n_2)$ precisely ensures that there is  a  solution for $n\in \mathcal{R}_{N}$, by use of  a slight generalization of the Chinese remainder theorem.  The task is to find  $\ell_1,\,\ell_2\in\Z$ such that $n={m_1}+\ell_1 n_1 = {m_2}+\ell_2n_2$ and by B\'ezout's identity the solution of this congruence exists if and only if $g=\gcd(n_1, n_2)$ divides ${m_1}-{m_2}$. The proposition then follows.
\end{proof}

It will be convenient to remark the following consequence of this result:

\begin{corollary}\label{cyclop}
 Let $\zeta_1$, $\zeta_2$ and $P$ be as in Proposition~\ref{cyclo}. If $\gcd(n_1,n_2)=1$ and the degree of $P$ in the first variable is less than $\deg C_{n_1}$ then $P\big(x,\zeta_2^{e_2}\big)$ is identically zero for every $e_2\in\mathcal{R}_{n_2}$.
\end{corollary}

\begin{proof}
 In view of Lemma \ref{cyclo2} and Proposition \ref{cyclo}, this corollary is almost obvious. The cyclotomic polynomial $C_{n_1}(x)$ divides $P\big(x,\zeta_2^{m_2}\big)$ because $x=\zeta_1^{m_1}$ with $m_1\in\mathcal{R}_{n_1}$ are zeros of the latter by Proposition~\ref{cyclo}.  This clearly leads to a contradiction if $P$ is not identically zero and $\deg C_{n_1}(x)>\deg P\big(x,\zeta_2^{e_2}\big)$.
\end{proof}

\section{Plateaux for the non fragmentation case}

To study the location of the plateaux for the non fragmentation case
it is convenient to introduce the function $D$ giving the distance to the nearest integer, defined by
\[
 D(x)=\textrm{dist}(x,\Z)=\min_{n\in\Z} |x-n|.
\]
This definition implies that $D(x)=D(1-x)$, a property that will be applied below.

\begin{theorem}\label{nofr}
 Assume that there is not fragmentation and that $2N\Lambda$ is an odd integer. Then there is a plateau given by the interval $[c-r,c+r]\cap [0,1/2]$ where  
 \[
  c = D\Big(\frac {2a}q N\Lambda+\frac 12\Big)
  \qquad\text{and}\qquad
  r = 
  \begin{cases}
  \frac{1}{q}D\big(\frac{q}{2\Lambda}+\frac 12\big)
  &\text{if $q$ is odd},
  \\[4pt]
  \frac{2}{q}D\big(\frac{q}{4\Lambda}+\frac 12\big)
  &\text{if $q$ is even}.
  \end{cases}
 \]
 Furthermore, it is a $0$-level plateau if and only if $q$ divides $4N\Lambda$.
\end{theorem}

\begin{proof}[Proof of Theorem~\ref{nofr}]
The formula for $q$ implies $qc=\pm\big(2aN\Lambda+\frac12 q-qn_0\big)$ for some choice of the sign and some $n_0\in\Z$. Consider first the situation with the $+$~sign, in this case it will be shown below that  $S_{-}(c)=0$. With a similar argument, the~$-$~sign gives $S_{+}(c)=0$. In fact, it is not necessary to repeat the calculations. Instead, if one uses the symmetry $a\mapsto -a$, $n_0\mapsto 1-n_0$ that flip the sign of $2aN\Lambda+\frac12 q-qn_0$ and transforms $S_{-}(c)$ into the conjugate of $S_{+}(c)$, it is clear than only one case needs to be considered.
 
 Unwrapping the definition in Proposition~\ref{criterion},
 $S_{-}(c)=0$ is equivalent to 
 \begin{equation}\label{summinus}
  \sideset{}{'}\sum_{k\in I(c)}
  e\big(P(k)\big)
  =0
  \qquad\text{with}\quad 
  P(k)=
  \begin{cases}
     \frac 1q\big(\overline{4a}k^2-N\Lambda k\big) &\text{if $q$ is odd},
     \\
     \frac 1{4q}\big(\overline{a}k^2-4N\Lambda k\big) &\text{if $q$ is even},
  \end{cases}
 \end{equation}
 where the prime indicates that the sum is restricted to $k+\frac{q}{2}$ even if $q$ is even.
 
 Let us use the notation $A\sim B$ in order to indicate that  $A-B\in\Z$. It is direct to check that $P(k\pm q)\sim P(k)+\frac 12$ which is  a simple consequence that $N\Lambda$ is half-integer. In addition
 \[
  P(2qc-k)
  =
  P(4aN\Lambda+q-2qc_0-k)
  \sim
  P(4aN\Lambda-k)+\frac 12.
 \]
 An  elementary calculation shows that 
 $
 P(4aN\Lambda-k)-P(k)\in\Z$  and all these facts allow to conclude that   $P(2qc-k)\sim P(k)+\frac 12$. 
 
 On the other hand, the interval $I(c)
 =\Big[
 -\frac{q}{2\Lambda}+cq,
 \frac{q}{2\Lambda}+cq
 \Big]\cap \Z$ is invariant under the action  $k\mapsto 2qc-k$. Note that $P(2qc-k)\sim P(k)+\frac 12$ implies that
 $e(P(2qc-k))=-e(P(k))$. Therefore the sum in \eqref{summinus} vanishes as the terms cancel in pairs. 
 
 By Proposition~\ref{criterion} there is a plateau  around $x=c$. 
 To assure that  
 $c$ is nonsingular, note that $qc$ is a half-integer if $q$ is odd and $qc+\frac12 q$ is an odd integer if $q$ is even while $\pm \frac{q}{2\Lambda}=\pm 2\frac{Nq}{2\Lambda N}$ is not a half-integer, because $2\Lambda N$ is odd, and it is not an odd integer for $q$ even. 
 
 Once $S_{-}(c)=0$ was established, given $q$  odd and $t>0$ satisfying
 \begin{equation}\label{inter}
  \Big(\big[ q(c-t)-\frac{q}{2\Lambda}, q(c+t)-\frac{q}{2\Lambda}\big]
  \cup
  \big[ q(c-t)+\frac{q}{2\Lambda}, q(c+t)+\frac{q}{2\Lambda}\big]\Big)\cap\Z
  =\emptyset
 \end{equation}
 then $I(c-t)$ and $I(c+t)$ capture the same integers and $S_{-}(x)$ remains invariant and null for $x\in [c-t,c+t]$. As $qc$ is a half-integer, this condition can be rephrased saying that the intervals $\big|x-\big(\frac 12-\frac{q}{2\Lambda}\big)\big|\le qt$ and $\big|x-\big(\frac 12+\frac{q}{2\Lambda}\big)\big|\le qt$ do not contain integers, which is equivalent to $$qt\le D\big(\frac 12-\frac{q}{2\Lambda}\big), \qquad \text{and}\qquad  qt\le D\big(\frac 12+\frac{q}{2\Lambda}\big).$$Using the obvious fact that $D(x)=D(1-x)$ and taking the supreme value in $t$, it is deduced  from Proposition~\ref{criterion} that the probability density is constant in $[c-r,c+r]\cap [0,\frac{1}{2}]$ with $r$ as in the statement. It is not possible to go beyond because replacing $r$ by $r'$ slightly bigger, each interval in \eqref{inter} contains exactly an integer point and hence $S_{-}(c)=0$ and $S_{-}(c\pm r')$ differ in a term of absolute value $1$. 
 
 The computation of $r$ in the even case can be treated in the analogous way by noticing that with the change of variables $2\ell=k+\frac{1}{2}q$ the prime in \eqref{summinus} can be dropped replacing the range of summation by $\ell\in\frac 12 I(c)+\frac{1}{2}q$. In this way, \eqref{inter} becomes
 \[
  \Big(\big[ \frac q2(c-t)-\frac{q}{4\Lambda}+\frac q2, \frac q2(c+t)-\frac{q}{4\Lambda}+\frac q2\big]
  \cup
  \big[ \frac q2(c-t)+\frac{q}{4\Lambda}+\frac q2, \frac q2(c+t)+\frac{q}{4\Lambda}+\frac q2\big]\Big)\cap\Z
  =\emptyset
 \]
 which is formally the same as before but replacing $q$ by $\frac q2$ and $c$ by $c+1$. As $\frac q2(c+1)$ is half-integer, the upper bounds of $t$ in the odd case follow replacing $q$ by $\frac q2$, giving the expected formula for~$r$. 
 \medskip 
 
 For the last claim in the statement, according to Proposition~\ref{criterion}, the plateau is a $0$-level plateau if and only if $S_{+}(c)=S_{-}(c)=0$. As pointed out at the beginning of the proof, the equality
 $qc=\pm\big(2aN\Lambda+\frac12 q-qn_0\big)$ implies that $S_{\mp}(c)=0$. Then $S_{+}(c)=S_{-}(c)=0$ if and only if both signs can be chosen in this formula for $qc$. In other words, if and only if  there exists $n_0,n_0'\in\Z$ such that 
 \[
  2aN\Lambda+\frac12 q-qn_0
  =
  -2aN\Lambda-\frac12 q+qn_0',
 \]
 that can be rephrased as $4aN\Lambda=q(n_0-n_0'-1)$ or simply as $q$ divides $4N\Lambda$ because $q$ and $a$ are coprime. This concludes the proof. 
\end{proof}

In fact, the intersection with $[0,1/2]$ (clipping the interval) is only needed in Theorem~\ref{nofr} for the $0$-level plateaux. They appear in the extremes: In the right extreme if $q$ is odd and in the left extreme if $q$ is even. This can be derived from the formulas for $c$ and $r$. 
\medskip

The computer simulations, suggests that Theorem~\ref{nofr} covers all the cases in which plateaux appear for $\Lambda\in\Q$.
\begin{conjecture}\label{existence}
 Assume that there is not fragmentation and $\Lambda\in\Q$.
 There exists a plateau if and only if $2N\Lambda$ is an odd integer. Moreover, it is unique. 
\end{conjecture}

We have been able to prove it for $q$ prime. 

\begin{theorem}\label{existence_p}
 If $q$ is prime, Conjecture~\ref{existence} is true.  
\end{theorem}

We state another partial result towards Conjecture~\ref{existence} for composite numbers.

\begin{theorem}\label{qcomp}
 If $q$ is composite with no repeated prime factors  and $2N\Lambda$ is an even integer then there are not plateaux. 
\end{theorem}

\medskip

Before going to the proof of these theorems, it is convenient to make some useful simplifications.

\

\emph{Preliminary simplifications:} 
In the following the number $q$ will be denoted $q=p$ to emphasize that is prime. In addition, the denominator of $N\Lambda$ will be denoted $s\in\Z^+$. Assume first that $p\nmid s$, the reason for this assumption will be clearly soon. Then 
\[
 \frac{\overline{4a}k^2\pm N\Lambda k}{p}
 =
 \frac{\overline{4a}k^2}{p}
 +
 \frac{r}{ps}k
 =
 \frac{\overline{4a}k^2}{p}
 +
 \frac{\alpha}{p}k
 +
 \frac{\beta}{s}k
\]
with $\frac{r}{s}$ an irreducible fraction, $\gcd(r,p)=1$. The values $p$ and $\alpha,\beta\in\Z$ with $\gcd(s,\beta)=1$ are determined by B\'ezout identity $r=\alpha s+\beta p$. 
After an integral translation $k\mapsto k+n_0$ it can be assumed
$I(x)\cap \Z=\{0,1,2,\dots, K\}$ with $0<K<p$. Note that this translation does not change the denominator $s$, only the values of $\alpha$ and $\beta$, whose names we keep. Hence, except for a constant nonzero factor, $S_+$ and $S_{-}$ are represented by a sum of the form 
\begin{equation}\label{Sform}
 \sum_{k=0}^K
  e\Big(
  \frac{\overline{4a}k^2+\alpha k}{p}
  +
  \frac{\beta k}{s}
  \Big)
  \qquad 
  \text{with}\quad
  0\le\alpha<p,\quad \gcd(\beta,s)=1
\end{equation}
A simplification arises now by use of Proposition~\ref{cyclo} with $n_1=p$, $n_2=s$, $m_1=4a$ and $m_2=1$.
To see in detail how this proposition is applied consider the following two variable polynomial with integer, thus rational, coefficients
$$ 
P(x,y) = \sum_{k=0}^K x^{\overline{4a}k^2+\alpha k}  y^{\beta k}.
$$
The choice $x=\zeta_1=e(\frac{1}{p})$ and $y=\zeta_2=e(\frac{1}{s})$ converts this polynomial into the  sum $S$.
 The  vanishing of $S$ is equivalent to the condition $P(\zeta_1,\zeta_2)=0$. The Proposition \ref{cyclo} assures the vanishing 
 $$
 P\big(\zeta_1^{m_1},\zeta_2^{m_2}\big)=0,
 $$
 for any  ${m_1}<p$  and ${m_2}<s$ with with $\gcd(m_1, p)=\gcd(m_2,s)=1$ satisfying that ${m_1}-{m_2}$ is a multiple of $\gcd(p,s)$. As $p$ is prime $\gcd(p,s)=1$ and the last condition is easily satisfied. The choice $m_1=4a$ is valid since $p$ is an odd prime, thus $\gcd(4a, p)=1$. With this choice and by taking $m_2=1$ it is deduced from the vanishing of \eqref{Sform} that
 \begin{equation}\label{simplifica}
 S
 =
 \sum_{k=0}^K
  e\Big(
  \frac{k^2+\alpha k}{p}
  +
  \frac{\beta k}{s}
  \Big)
  \end{equation}
vanishes, where the value of $\alpha$ has of course been modified. The simplification that follows from \eqref{simplifica} is that the sum $S$ with $\overline{4a}$
 gives the same information  that the one with $\overline{4a}=1$. In the following, the expression \eqref{simplifica} will be used, as it is easier to deal with.

Finally, let us justify the assumption $p\nmid s$ showing  that the case  $p\mid s$ cannot give a vanishing sum. To see this, note that $s$ may be decomposed in this case as $s=p^ut$ with $p\nmid t$ and $u\in\Z^+$. Now use the translation employed for $p\nmid s$ above, but without using  B\'ezout identity. The task is to rule out
\begin{equation}\label{asumiendo}
 \sum_{k=0}^K
  e\Big(
  \frac{\overline{4a}k^2+\alpha k}{p}
  +
  \frac{r k}{p^{u+1}t}
  \Big)=0
  \qquad\text{with}\quad 
  \frac{r}{p^ut}\text{ irreducible, } u\in\Z^+.
\end{equation}
The trick is now to apply the Proposition~\ref{cyclo}, as for the previous case, but with the choice $n_1=p$, $n_2=p^{u+1}t$, $m_1=1$ and $m_2=m_2(j)=1+j\overline{t} t p^u$ with $\overline{t}$ the inverse of $t$ modulo $p$ and $j$ ranging from $0$ to $p-1$. 
Note that the assumptions of the Proposition \ref{cyclo} require that $$\gcd(m_2, p^{u+1}t)=1.$$ This condition follows from the fact that
 $m_2 -j\overline{t} t p^u=1$ which leads to the conclusion that $m_2$ can not be divided by $p$ or any other nontrivial factor of $t$. In addition, as $\gcd(n_1,n_2)=p$ the other requirement is that $m_1-m_2=1-m_2$ is divisible by $p$, which is clearly true.

Then the Proposition \ref{cyclo}
can be applied for every $j$. The assumption \eqref{asumiendo} together with the Proposition \ref{cyclo} leads to
$$
\sum_{k=0}^K \zeta_1^{\overline{4a}k^2+\alpha k}\zeta_2^{r k}
  e\Big(\frac{r\overline{t} k j}{p}\Big)=0.
$$
By summing over $j$ it follows that
\[
  \sum_{j=0}^{p-1}
  \sum_{k=0}^K \zeta_1^{\overline{4a}k^2+\alpha k}\zeta_2^{r k}
  e\Big(\frac{r\overline{t} k j}{p}\Big)
  =
  \sum_{k=0}^K \zeta_1^{\overline{4a}k^2+\alpha k}\zeta_2^{r k}
  \sum_{j=0}^{p-1}
  e\Big(\frac{r\overline{t} k j}{p}\Big)=0.
\]
However, as $K<p$,  this result can not be true since the innermost sum equals $p$ for $k=0$ and it vanishes in the rest of the cases. 
This contradiction shows that the only case to be considered is $p\nmid s$.
\medskip

After the previous simplification, the following proofs of the theorems can be presented.
\medskip

\begin{proof}[Proof of Theorem~\ref{existence_p}]
 After all of the previous reductions, it follows that the vanishing of $S_{\pm}(x)$ translates into the vanishing of a sum $S$ identical to that of \eqref{Sform} but changing $\overline{4a}$ by $1$.  Recall that $s$ is the denominator of $\Lambda N$ and $p\nmid s$. Corollary \ref{cyclop}
 cannot be applied directly by defining a function of the form
 \[
  P(x)=\sum_{k=0}^K 
  x^{k^2+\alpha k}e\Big(\frac{k}{s}\Big).
 \]
 This polynomial function $P(x)$ reduces to \eqref{simplifica} if $x=e(\frac{1}{p})$. However, a requirement of the Corollary \ref{cyclop} is that the function $f(k)=k^2+\alpha k$ is such that $0<f(k)<\text{deg}\;C_p=p-1$, which is not assured here. This technical problem can be avoided if one realizes the freedom to choose $f(k)$ modulo $p$. One may try to define $g(k)$ as the remainder of $f(k)$ when divided by $p$. However, one needs this remainder to be less than $p-1$, and this is not assured, as the remainder may be equal to $p-1$. On the other hand, as this does not cover all the values one may shift them to be sure  that the value $p-1$ is not reached.

In order to be more concrete, the function $f(x)= k^2+\alpha k$ is not bijective modulo $p$ (note that $f(0)=f(-\alpha)$ for $\alpha\ne 0$). So, we can find a value $b\not\in\text{Im}\, f$ that is, $b$ is a value such that $b\neq f(k)$ for any $k$. Define $g(k)$ as the remainder when $f(k)+p-1-b$ is divided by $p$. The choice of $b$ shows $g(k)\ne p-1$ therefore $0\le g(k)<p-1$. 
 In order to exemplify this, consider  $p=5$. Then the image of  $f(x)=x^2+3x$ modulo 5 does not contain $1$. By defining  $g(x)$ as the remainder of  $x^2+3x+4-1$ it can not be  $4$, since otherwise the remainder of  $x^2+3x$ will be equal to $1$. The value of $b$ in this case is precisely $b=1$, which is not in the image of $f(k)$.

 If $S=0$, then by applying Corollary~\ref{cyclop} with $n_1=p$, $n_2=s$ and $e_2$ the inverse of $\beta$ modulo $s$ it is deduced 
 \[
  \sum_{k=0}^K 
  x^{g(k)}e\Big(\frac{k}{s}\Big)
  =
  0.
 \]
 As $g$ is given by a quadratic function, each value in its image can be taken~$1$ or~$2$ times. The first case is impossible because if $g(k_0)$ is taken only once, the coefficient of $x^{g(k_0)}$ would be $e(k_0/s)\ne 0$. In the second case, if $g(k_0)=g(k_1)$ with $k_0\ne k_1$, the coefficient of $x^{g(k_0)}$ is
 \begin{equation}\label{cancels}
  e\Big(\frac{k_0}{s}\Big)+e\Big(\frac{k_1}{s}\Big)=0.
 \end{equation}
Furthermore, given a quadratic equation $x^2+ax+b$ modulo $p$ with $k_0$ and $k_1$ its solutions, then $x^2+ax+b$ and $(x-k_0)(x-k_1)$
are equal modulo $p$. By comparing linear coefficients, it follows that $a$ and $-k_0-k_1$ are equal modulo $p$. This discussion shows that $k_0+k_1=\ell p-\alpha$ with $\ell$ some integer. In particular, for $k_0=0,\, 1$, 
 \[
  1+e\Big(\frac{\ell p-\alpha}{s}\Big)=0
  \qquad\text{and}\qquad
  e\Big(\frac{1}{s}\Big)+ e\Big(\frac{\ell p-\alpha}{s}\Big)e\Big(-\frac{1}{s}\Big)=0. 
 \]
 This implies $e(\frac{2}{s})=1$ and $s=2$. Recalling that $s$ is the denominator of $\Lambda N$ it is deduced that $2\Lambda N$ is an odd integer. On the other hand, Theorem~\ref{nofr} assures that actually in this case there are plateaux. 
 
 Finally, to deduce the claimed uniqueness, we are going to show that the existence of two plateaux leads to a contradiction. 
 Consider first the case in which they have the same character in Proposition~\ref{criterion}: 
 $S_{+}$ vanishes in both plateaux
 or
 $S_{-}$ vanishes in both plateaux.
 Substituting $s=2$ in the initial reductions, which forces $\beta$ to be odd, the existence of two plateaux implies
 \[
  \sum_{k=0}^K
  e\Big(
  \frac{k^2+\alpha k}{p}
  +
  \frac{k}{2}
  \Big)
  =
  \sum_{k=K'}^{K''}
  e\Big(
  \frac{k^2+\alpha k}{p}
  +
  \frac{k}{2}
  \Big)
  =0
 \]
 for some $0<K<K'<K''<p$. According with the previous argument, the only way of getting the vanishing of these sums is that the terms in each of them cancel in pairs satisfying $k_0+k_1=\ell p-\alpha$ and this must be odd to fulfill \eqref{cancels} with $s=2$. 
 As $k_0,k_1,\alpha\in [0,p-1]$, this determines $\ell$ completely: $\ell=1$ if $2\mid \alpha$ and $\ell=2$ if $2\nmid\alpha$. Summing up, the intervals $[0,K]$ and $[K',K'']$ should be both symmetric with respect to $(\ell p-\alpha)/2$ and it is impossible because they do not overlap. 
 
 Now consider the case in which both plateaux have different character: $S_{+}$ vanishes only in one of them and $S_{-}$ only in the other. 
 They cannot invade the boundary of $[0,1/2]$ because it would contradict the assumption with a simultaneous vanishing of $S_{+}$ and $S_{-}$ there (by the last part of Proposition~\ref{criterion} and the vanishing of 
 the probability density at $x=0,1/2$ by the symmetries of $\Psi$). Let $c_+$ and $c_-$ be the centers of the two plateaux. Applying the reductions without the translation, 
 \[
  \sum_{k\in I(c_+)}
  e\Big(
  \frac{k^2+\alpha k}{p}
  +
  \frac{k}{2}
  \Big)
  =
  \sum_{k\in I(c_-)}
  e\Big(
  \frac{k^2-\alpha k}{p}
  +
  \frac{k}{2}
  \Big)
  =0.
 \]
 The center of $I(c)$, its axis of symmetry, is $qc$. Hence, 
 $pc_{\pm}=\frac12 (\ell_{\pm} p\mp\alpha)$
 for some $\ell_+,\ell_-\in\Z$ with $\ell_{\pm} p\mp\alpha$ odd. This gives $c_++c_-\in\Z$ and contradicts $0<c_+,c_-<1/2$. 
\end{proof}

\begin{proof}[Proof of Theorem~\ref{qcomp}]
 Say that $p_1p_2\cdots p_n$ is the factorization in distinct primes of $q$ and write $q_j=q/(p_1\cdots p_j)$ for $1\le j\le n$.
 
 Assume first that $q$ is odd. Applying the preliminary reductions leading to \eqref{simplifica} with $s=1$ (because $N\Lambda\in\Z$), the nonexistence of plateaux is deduced proving that 
 \[
  S_0=\sum_{k\in \mathcal{C}_0}
  e\Big(\frac{k^2+\alpha k}{q}\Big)
  \qquad\text{with}\quad
  \mathcal{C}_0=\{0,1,2\dots, K\}
 \]
 cannot be zero.
 Define for $1\le j\le n$
 \[
  S_j=\sum_{k\in \mathcal{C}_j}
  e\Big(\frac{k^2+\alpha k}{q_j}\Big)
  \qquad\text{with}\quad
  \mathcal{C}_j=\big\{k\in \mathcal{C}_{j-1}\,:\, p_j\text{ divides }k^2+\alpha k\big\}. 
 \]
 Note that $0\in \mathcal{C}_j$, so they are nonempty. 
 Proving that $S_{j-1}=0$ implies $S_j=0$ we are done because $S_0=0$ would give inductively the contradiction $S_n=\sum_{k\in \mathcal{C}_n}1=0$. 
 
 Using B\'ezout's identity, there exist $a,b\in\Z$ with $\gcd(a,p_j)=\gcd(b,q_j)=1$ and $aq_j+bp_j=1$. Then 
 \[
  S_{j-1}=
  \sum_{k\in \mathcal{C}_{j-1}}
  e\Big(a\frac{k^2+\alpha k}{p_j}\Big)
  e\Big(b\frac{k^2+\alpha k}{q_j}\Big).
 \]
 If $S_{j-1}=0$, by Proposition~\ref{cyclo} with $m_1$ and $m_2$ the inverses of $a$ and $b$ modulo $p_j$ and $q_j$, the sum still vanishes when $a$ and $b$ are omitted. As we saw in the proof of Theorem~\ref{existence_p}, there exists $0\le b_j<p_j$ such that $k^2+\alpha k +p_j-1-b_j$ has remainder $g(k)<p_j-1$ when divided by $p_j$ and by Corollary~\ref{cyclop} it is deduced 
 \[
  \sum_{k\in \mathcal{C}_{j-1}}
  x^{g(k)}
  e\Big(\frac{k^2+\alpha k}{q_j}\Big)=0. 
 \]
 The vanishing of the coefficient of $p-1-b_j$ implies $S_j=0$ as expected because $g(k)=p_j-1-b_j$ if and only if $p_j$ divides $k^2+\alpha k$. This concludes the proof of the odd case. 
 \medskip
 
 The even case only differs in the first step. If $q$ is even, $q/2$ is odd and \eqref{coef} restricts $k$ to odd values. Proceeding as in \eqref{simplifica} but under this restriction and with the new formula of $c_{a/q}(k)$, we want to prove that it is not possible to have 
 \[
  \sum_{\substack{k=1\\ k\text{ odd}}}^L
  e\Big(\frac{k^2}{4q}+\frac{\alpha k}{q}\Big)=0
  \qquad\text{with}\quad
  L = 2K+1\ge 1. 
 \]
 Changing $k$ by $2k+1$, this can be re-written as 
 \[
  S_0=\sum_{k\in \mathcal{C}_0}
  e\Big(\frac{k^2+(4\alpha+1) k}{q}\Big)=0
  \qquad\text{with}\quad
  \mathcal{C}_0=\{0,1,2\dots, K\}
 \]
 and we can proceed as before. It is interesting to note that $k^2+(4\alpha+1) k$ is constantly zero modulo~$2$ and then the step dealing with the prime $2$ is rather trivial. 
\end{proof}

We are also able to deal with the extension of Conjecture~\ref{existence} to $\Lambda\not\in\Q$ proving the absence of plateaux for algebraic irrational values of $\Lambda$. 
Our proof depends on a celebrated deep result in number theory. 

\begin{theorem}\label{Lalgebr}
 If there is not fragmentation and $\Lambda$ is an irrational algebraic value then there are not plateaux. 
\end{theorem}

Before proving Theorem~\ref{Lalgebr}, recall that algebraic numbers are defined as the roots of polynomials with rational coefficients, they include all the nested expressions with radicals and the elementary operations,  for instance $\sqrt{2}$, $\sqrt{2}-\sqrt[4]{3}$  or $\sqrt[3]{7+\sqrt{2}}/\sqrt{1+3\sqrt{5}}$. After the work of Abel and Galois, it is known that there are algebraic numbers not admitting this kind of expressions. It is also known that a root of a polynomial with algebraic coefficients is also algebraic (this derives from elementary considerations about what is called the degree of a field extension). 
\medskip

Gelfond-Schneider theorem is one of the most celebrated results in transcendence theory, it solves Hilbert's 7th problem \cite{hille} (which he considered ``very difficult''). It states that if $\alpha\ne0,1$ and $\beta$ are (real or complex) algebraic numbers with $\beta\not\in\Q$, then $\alpha^\beta$ is not algebraic. 
Two examples mentioned by Hilbert himself are $\alpha=2$, $\beta=\sqrt{2}$ and $\alpha=e^{i\pi/2}$, $\beta=-2i$ showing that $2^{\sqrt{2}}$ and $e^\pi$ are not roots of a polynomial with rational (or even algebraic) coefficients. 

This deep result allows to extend Conjecture~\ref{existence} to irrational algebraic values.

\begin{proof}[Proof of Theorem~\ref{Lalgebr}]
 According to Proposition~\ref{criterion} to have a plateau
 $S_{+}(t)=0$ or $S_{-}(t)=0$ for some fixed~$t$. This means that $X_0=e\big(\pm \frac{N\Lambda}{q}\big)$ is a root of the polynomial $P(X)= \sum_{k\in I}c_{a/q}(k)X^k$ with $I=I(t)$. If there is not fragmentation, the sum is nonempty and $P$ is not identically zero. Its coefficients $c_{a/q}(k)$  are algebraic numbers because they satisfy $X^q-1=0$ if $q$ is odd and $X^{4q}-2^{2q}=0$ when $q$ and $q+k/2$ are even. Hence, $X_0$ is algebraic and it contradicts Gelfond-Schneider theorem taking $\alpha=e\big(\pm\frac{1}{2q}\big)$ and $\beta=2N\Lambda$.
\end{proof}

The set of algebraic irrational numbers is very wide but it does not include all the real irrational numbers (as shown in 1844 by Liouville). Then a loose end is to know if there is a counterexample dropping the word ``algebraic'' in Theorem~\ref{Lalgebr}. This is equivalent to ask if the polynomial $P$ in  the previous proof may have a root on the unit circle which is not a root of the unity. A famous example due to Lehmer implies that~$8$ out the~$12$ roots of $x^{12}-x^7-x^6-x^5+1$ satisfy this property. Examples like this  and some theoretical results \cite{daileda} seem to suggest that counterexamples to an extension of Theorem~\ref{Lalgebr} for non algebraic values of $\Lambda$ might exist, although we have not a clear conjecture about it. 

We have some minor results supporting other cases of Conjecture~\ref{existence} that we do not consider relevant enough to be reflected here. We think that some known results about linear relations between roots of the unity 
\cite{mann}
\cite{LaLe}
may play a role here.

\section{Numerical examples}

Let us first consider examples of the fragmentation case. Take $\Lambda=10.7$, for $q=7,\,12$ and $10$ there is fragmentation because $\Lambda>7$, 
$\Lambda>6$ and 
$\Lambda>5$, respectively. 
In Figure~\ref{fig_frag} it is shown the graph of $p$ for the choices
$a/q=2/7$, $N=1$;
$a/q=1/12$, $N=2$ and
$a/q=3/10$, $N=1$
which exemplify the three cases in 
Theorem~\ref{fragmen}.

\begin{figure}[H]
 \begin{center}
 \begin{tabular}{c}
  \includegraphics[height=85pt]{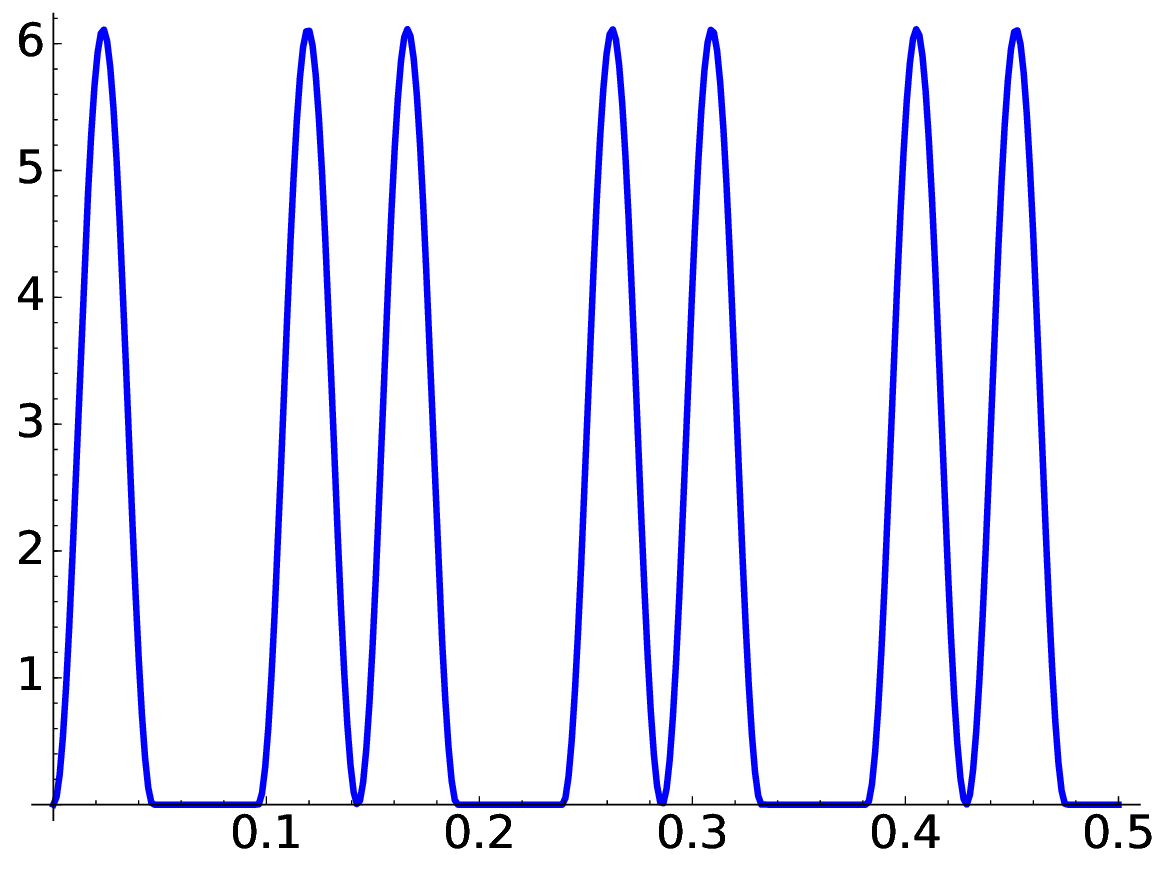}\hspace{-2pt}(a)
 \end{tabular}
 \ 
 \begin{tabular}{c}
  \includegraphics[height=85pt]{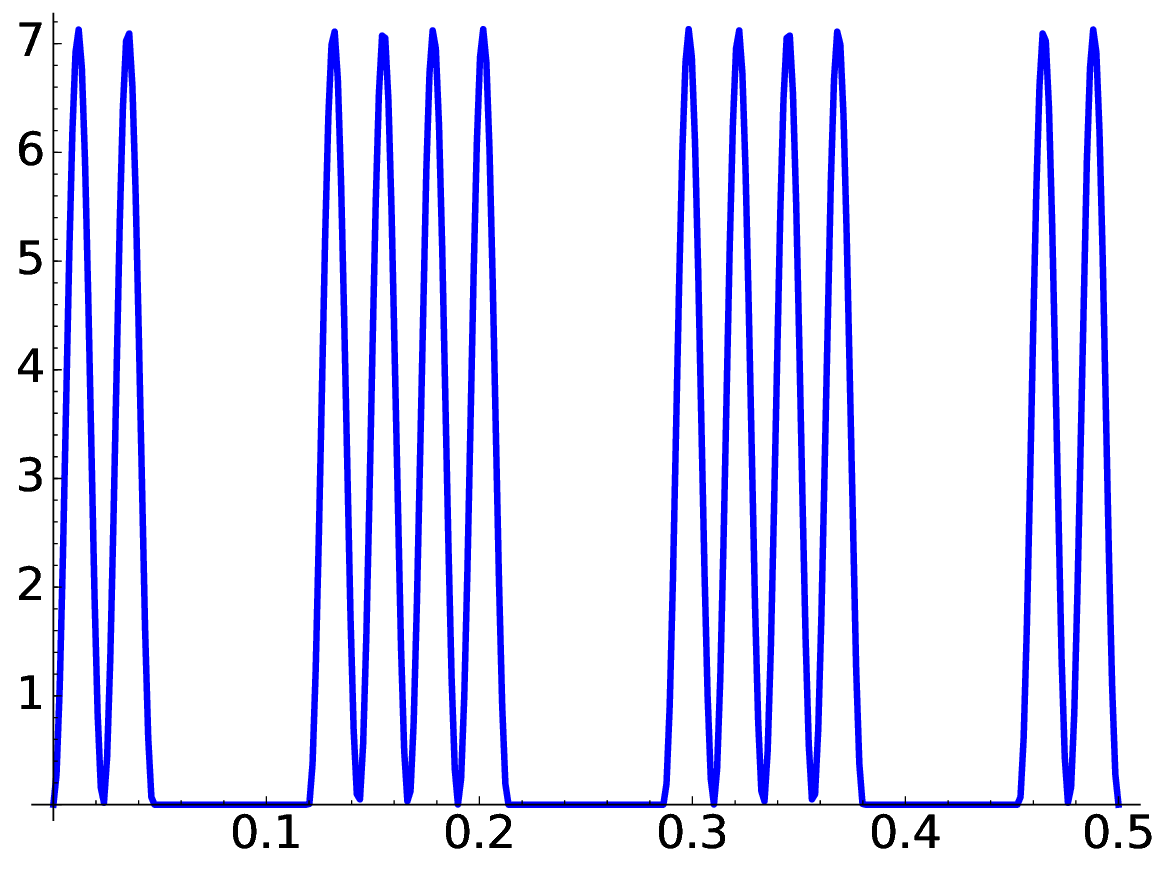}\hspace{-2pt}(b)
 \end{tabular}
 \ 
 \begin{tabular}{c}
  \includegraphics[height=85pt]{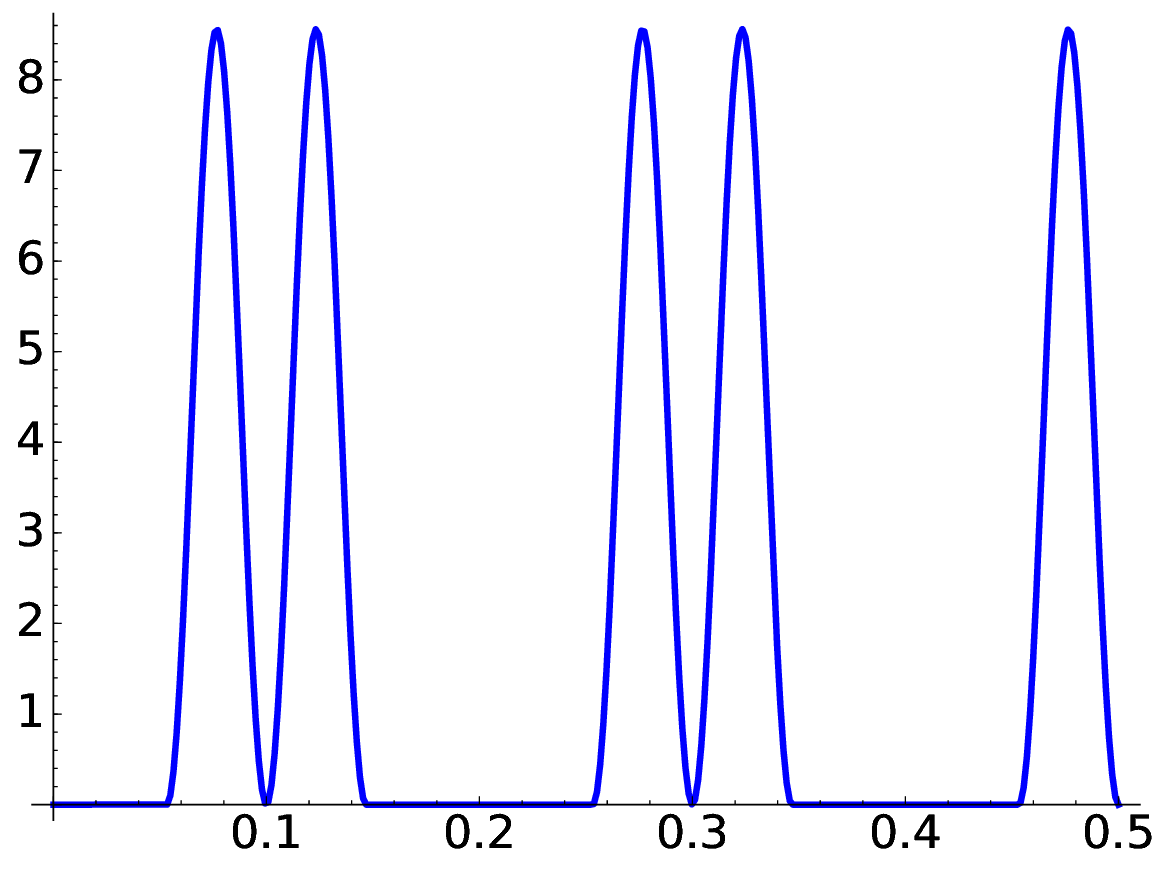}\hspace{-2pt}(c)
 \end{tabular}\vspace{-10pt}
 \end{center}
 \caption{Fragmentation case. 
 (a) $\Lambda=10.7$, $a/q=2/7$, $N=1$.
 (b) $\Lambda=10.7$, $a/q=1/12$, $N=2$.
 (c) $\Lambda=10.7$, $a/q=3/10$, $N=1$.
 }\label{fig_frag}
\end{figure}

In the rest of the examples, there is not fragmentation. In the figures, the center of the unique plateau is marked with a solid line and the boundaries in $(0,1/2)$ with dashed lines. 
\medskip

Consider 
$\Lambda=5/2$, $a/q=1/3$, $N=1$. According to Theorem~\ref{existence_p}, there is only a plateau. 
By Theorem~\ref{nofr}, its center and radius are
\[
 c= D\Big(\frac{13}{6}\Big)=\frac 16
 \qquad\text{and}\qquad
 r = \frac{1}{3}D\Big(\frac{11}{10}\Big)=\frac{1}{30}. 
\]
In Figure~\ref{nfig_frag1} they are also considered the cases 
$\Lambda=5/2$, $a/q=13/18$, $N=3$ and
$\Lambda=5/4$, $a/q=11/6$, $N=2$ with $q$ even.  In total, in these three cases the plateaux deduced from the formulas are 
\[
 \big[
 \frac{2}{15},\frac{1}{5}
 \big],
 \qquad
 \big[
 \frac{3}{10},\frac{11}{30}
 \big],
 \qquad\text{and}\qquad
 \big[
 \frac{7}{30},\frac{13}{30}
 \big],
\]
which agree with the plots. 

\begin{figure}[H]
 \begin{center}
 \begin{tabular}{c}
  \includegraphics[height=85pt]{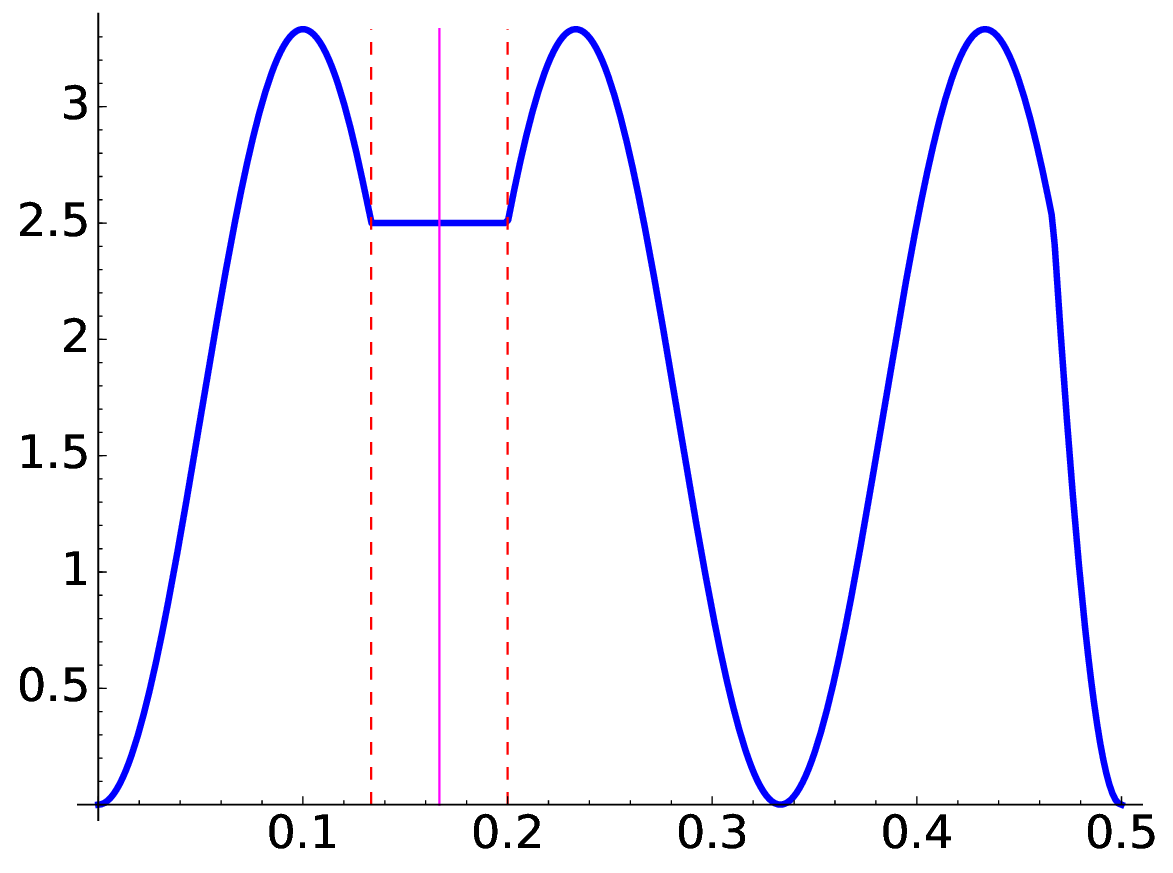}\hspace{-2pt}(a)
 \end{tabular}
 \ 
 \begin{tabular}{c}
  \includegraphics[height=85pt]{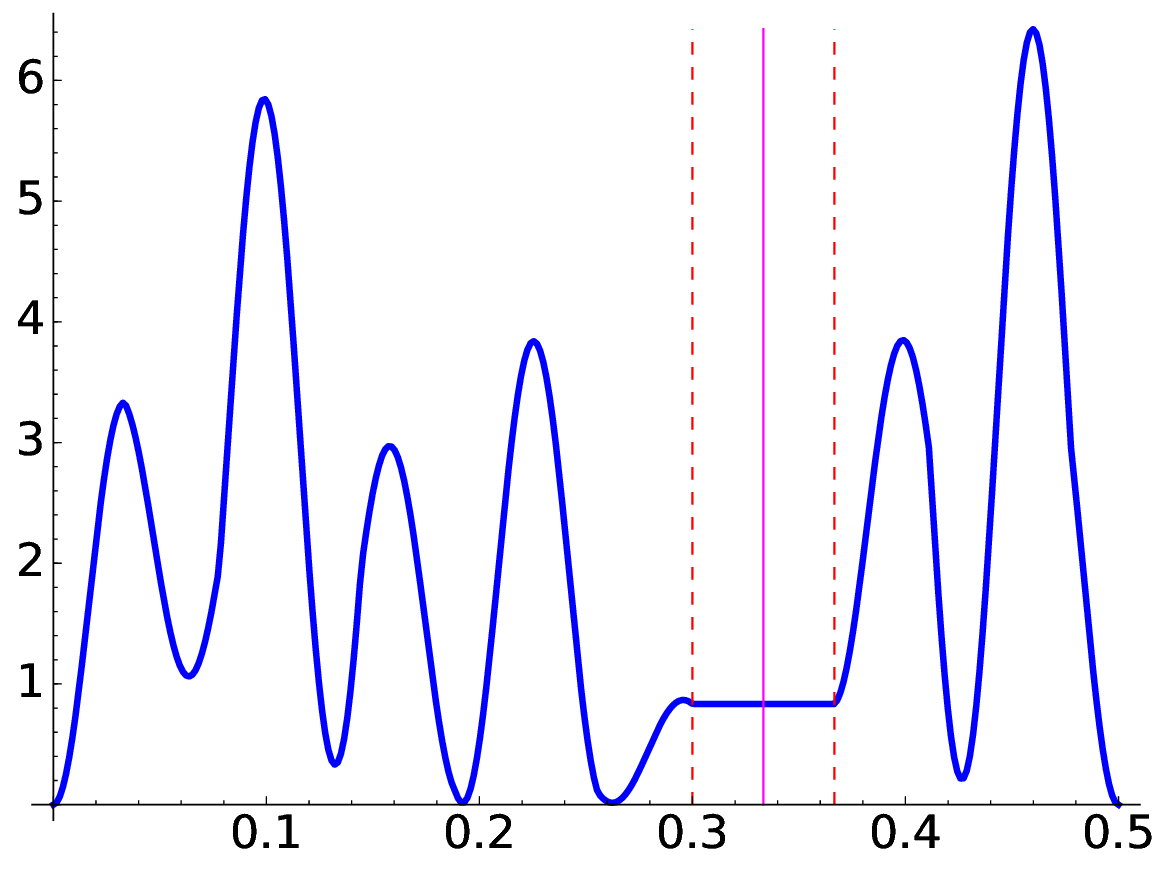}\hspace{-2pt}(b)
 \end{tabular}
 \ 
 \begin{tabular}{c}
  \includegraphics[height=85pt]{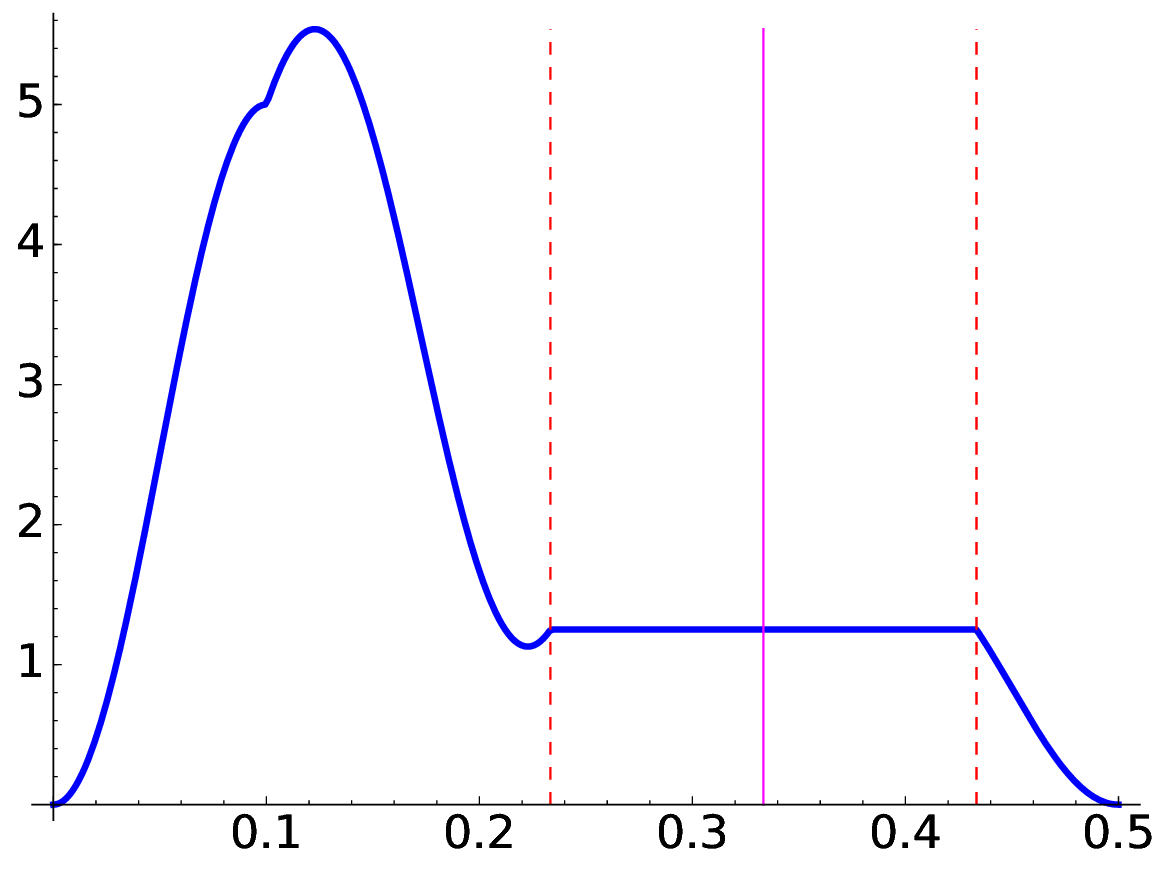}\hspace{-2pt}(c)
 \end{tabular}\vspace{-10pt}
 \end{center}
 \caption{No fragmentation and nonzero level plateaux. 
 (a) $\Lambda=5/2$, $a/q=1/3$, $N=1$.
 (b) $\Lambda=5/2$, $a/q=13/18$, $N=3$.
 (c) $\Lambda=5/4$, $a/q=11/6$, $N=2$.
 }\label{nfig_frag1}
\end{figure}

In Figure~\ref{nfig_frag2} there are more examples with $N>1$. As suggested by Lemma~\ref{bformula}, higher values of $N\Lambda$ give more oscillations.  

\begin{figure}[H]
 \begin{center}
 \begin{tabular}{c}
  \includegraphics[height=85pt]{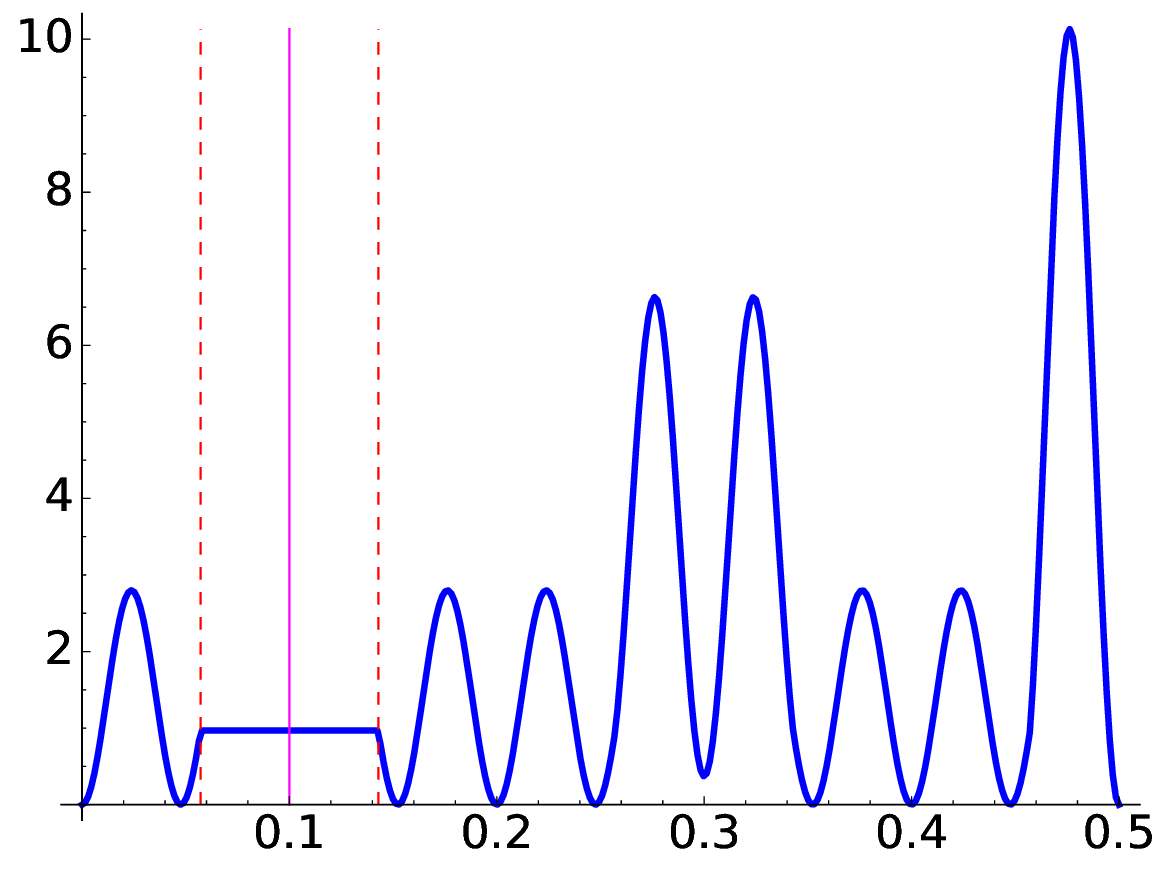}\hspace{-2pt}(a)
 \end{tabular}
 \ 
 \begin{tabular}{c}
  \includegraphics[height=85pt]{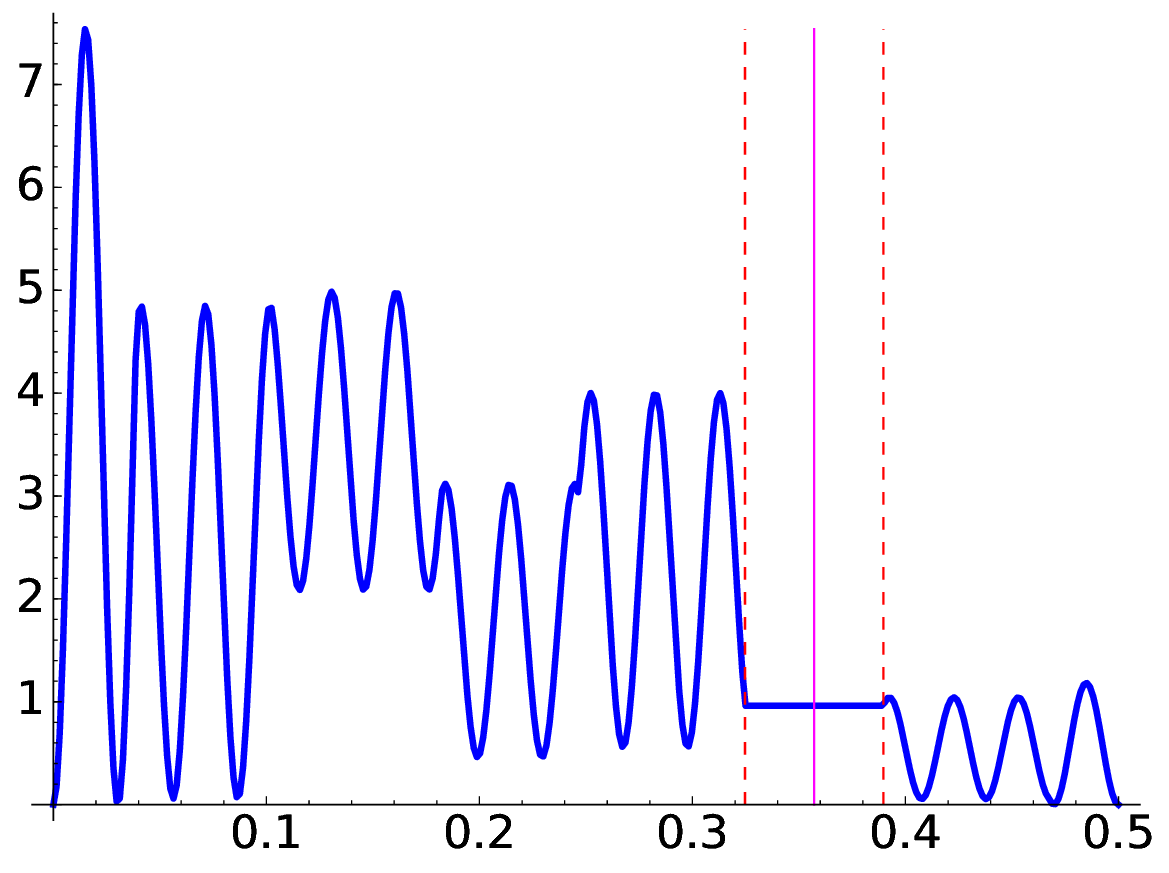}\hspace{-2pt}(b)
 \end{tabular}
 \ 
 \begin{tabular}{c}
  \includegraphics[height=85pt]{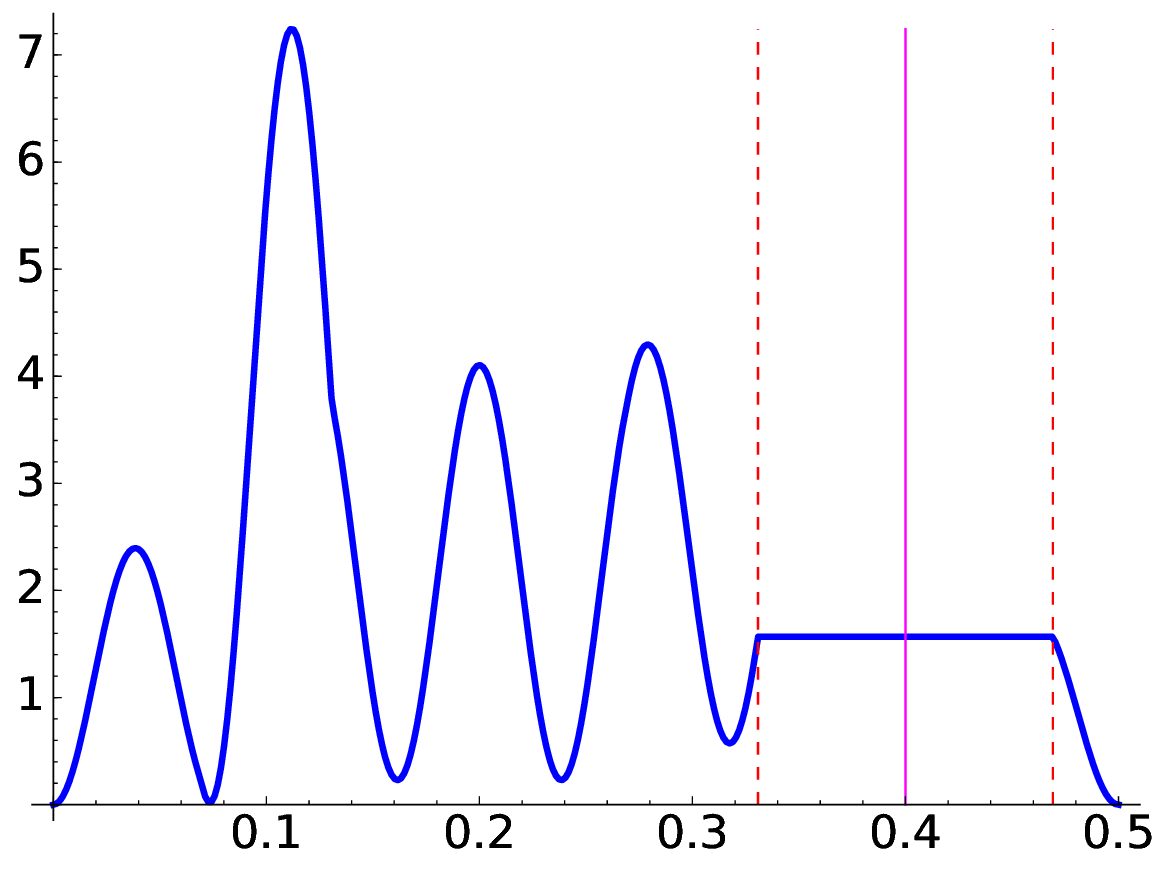}\hspace{-2pt}(c)
 \end{tabular}\vspace{-10pt}
 \end{center}
 \caption{More examples of no fragmentation and nonzero level plateaux.
 (a) $\Lambda=7/2$, $a/q=8/5$, $N=3$.
 (b) $\Lambda=11/4$, $a/q=3/7$, $N=6$.
 (c) $\Lambda=13/6$, $a/q=7/10$, $N=3$.
 }\label{nfig_frag2}
\end{figure}

Finally, in Figure~\ref{nfig_frags} it is illustrated the case of $0$-level plateaux without fragmentation. Fix $\Lambda =3/2$. Then the conditions in Theorem~\ref{nofr}
read $q$ divides $3N$ and $q$ divides $6N$. So, $q=3$ with $N=1$ is a valid choice and, for the even case, $q=6$, $N=3$ and $q=18$, $N=3$ are also valid. In every case, the argument of $D$ in the formula for $r$ is a half integer. On the other hand, $\frac {2a}q N\Lambda+\frac 12$ is half integer for $q=7$ an integer in the even cases. So, according to the formulas, the plateaux are 
\[
 \big[
 \frac{1}{3},\frac{1}{2}
 \big],
 \qquad
 \big[
 0,\frac{1}{6}
 \big],
 \qquad\text{and}\qquad
 \big[
 0,\frac{1}{18}
 \big].
\]
Note that the plateaux appears to the right of $[0,1/2]$ in the odd case and to the left in the even case. 

\begin{figure}[H]
 \begin{center}
 \begin{tabular}{c}
  \includegraphics[height=85pt]{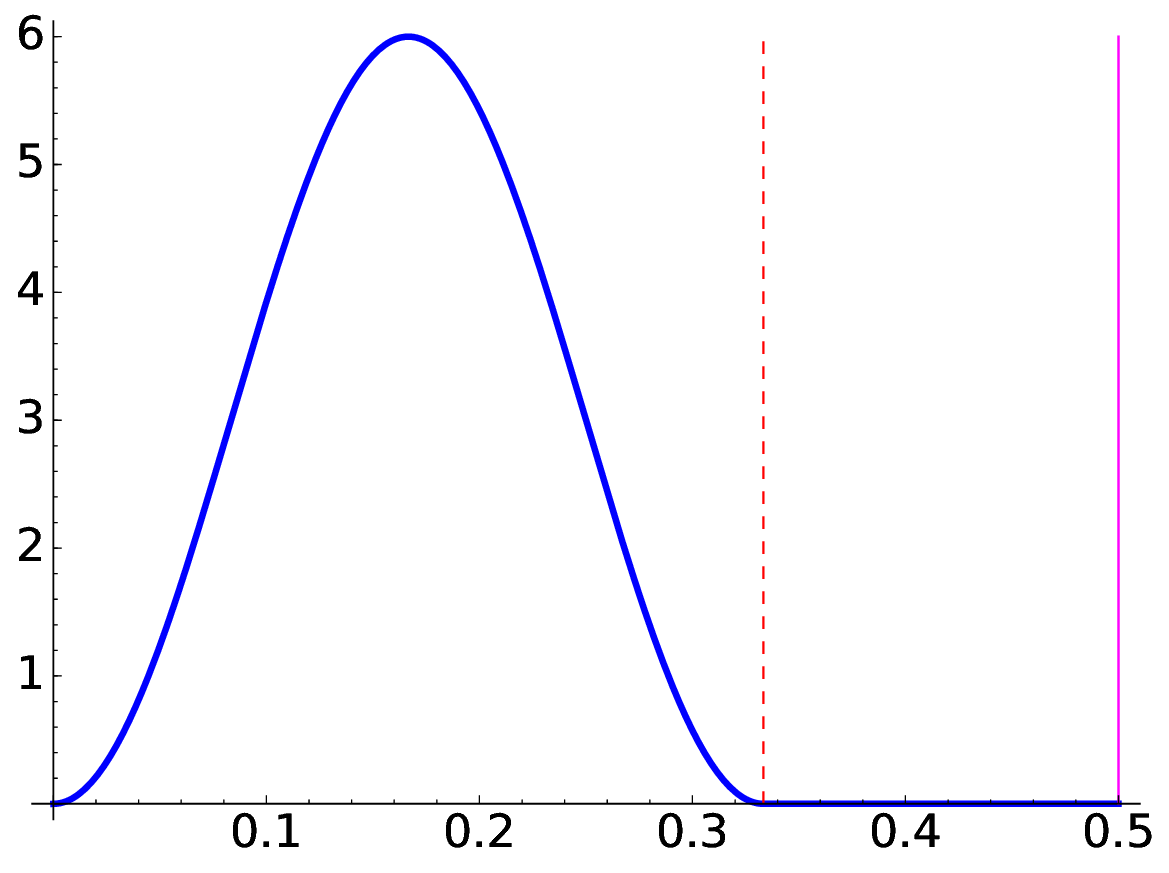}\hspace{-2pt}(a)
 \end{tabular}
 \ 
 \begin{tabular}{c}
  \includegraphics[height=85pt]{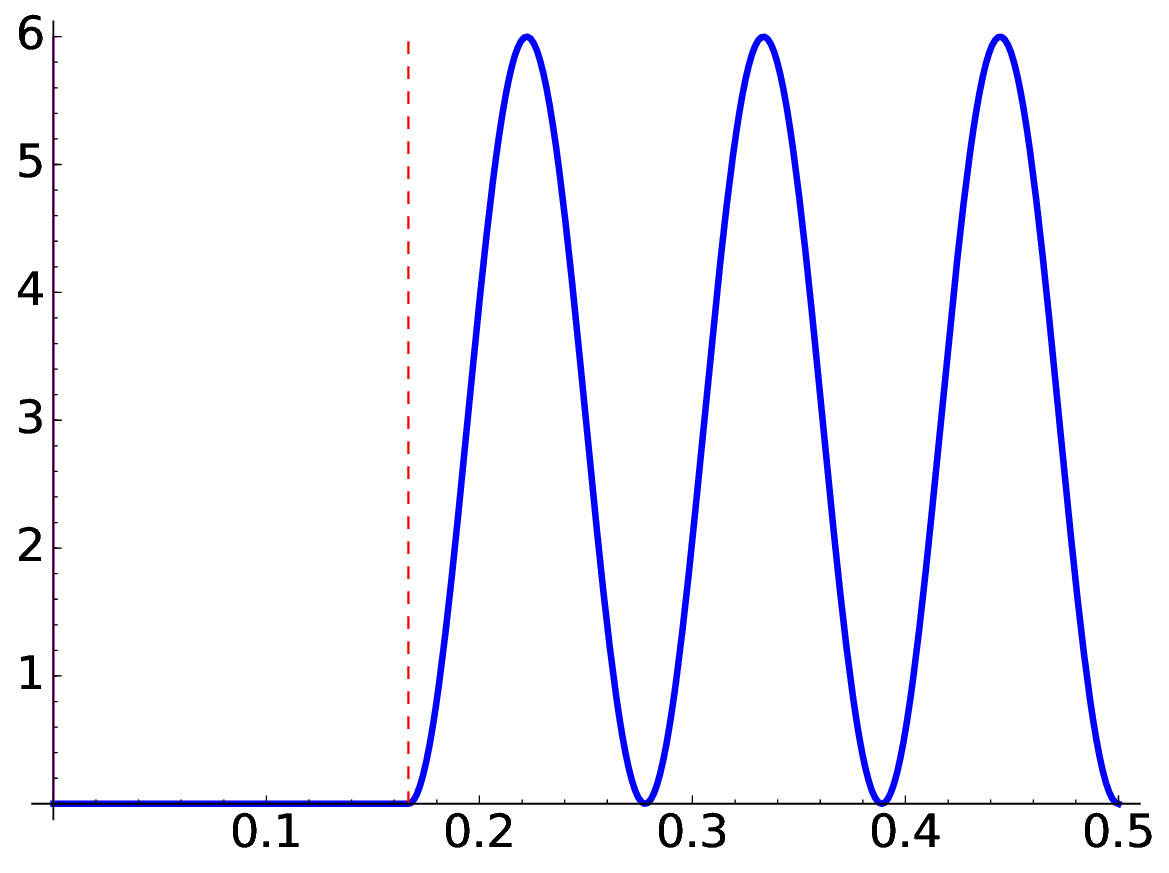}\hspace{-2pt}(b)
 \end{tabular}
 \ 
 \begin{tabular}{c}
  \includegraphics[height=85pt]{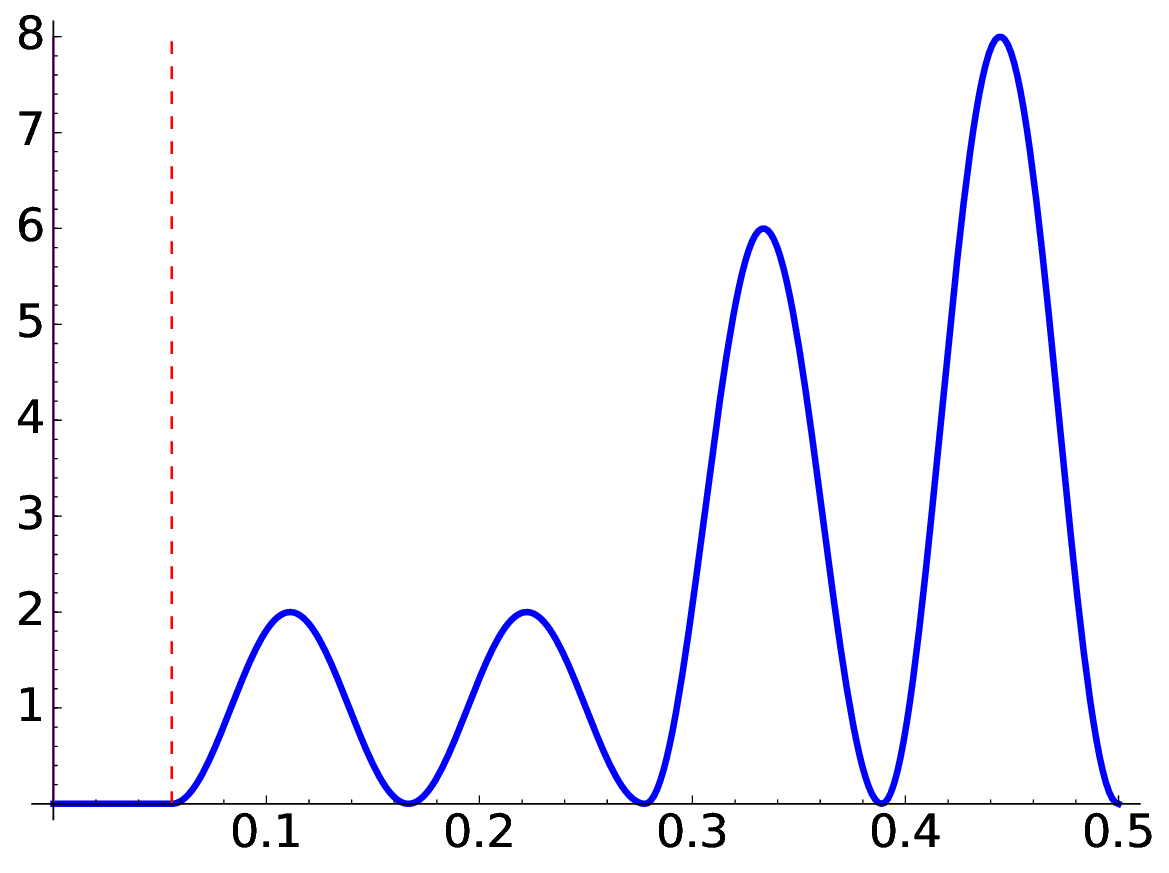}\hspace{-2pt}(c)
 \end{tabular}\vspace{-10pt}
 \end{center}
 \caption{$0$-level plateaux without fragmentation. 
 (a) $\Lambda=3/2$, $a/q=5/3$, $N=1$.
 (b) $\Lambda=3/2$, $a/q=1/6$, $N=3$.
 (c) $\Lambda=3/2$, $a/q=7/18$, $N=3$.
 }\label{nfig_frags}
\end{figure}

\section*{Acknowledgments}
O.P. S. is partially supported by CONICET, Argentina and by the Grant PICT 2020-02181. D. R. is partially supported by the PID2020-113350GB-I00 grant of the MICINN (Spain) and F. Ch. is partially supported by the latter grant and
by ``Severo Ochoa Programme for Centres of Excellence in R{\&}D'' (CEX2019-000904-S).


\end{document}